\documentclass[journal]{IEEEtran}

\usepackage{graphicx}

\usepackage[T1]{fontenc}
\usepackage{subcaption}
\usepackage{color}
\usepackage{varioref}
\usepackage{textcomp}
\usepackage{amsthm}
\usepackage{amsmath}
\usepackage{amssymb}
\usepackage{setspace}
\usepackage{esint}
\usepackage{algorithm}
\usepackage{algpseudocode}
\usepackage{pifont}
\usepackage{wrapfig}
\usepackage{multirow}
\usepackage{tabu}
\usepackage{amsmath}

\usepackage{enumitem}
\usepackage{cite}
\usepackage{cleveref}
\makeatletter

\newcommand{\Rmnum}[1]{\expandafter\@slowromancap\romannumeral #1@}
\makeatother

\newtheorem{theorem}{Theorem}
\newtheorem{lemma}{Lemma}
\newtheorem{remark}{Remark}
\newtheorem{definition}{Definition}
\newtheorem{assumption}{Assumption}
\newtheorem{corollary}{Corollary}
\theoremstyle{definition}

\providecommand{\propositionname}{Proposition}

\ifCLASSINFOpdf

\else

\fi

\usepackage[T1]{fontenc}
\usepackage{etoolbox}

\makeatletter
\patchcmd{\maketitle}{\@fnsymbol}{\@alph}{}{}  
\makeatother

\title{Machine Learning at the Wireless Edge: Distributed Stochastic Gradient Descent Over-the-Air}
\author{\IEEEauthorblockN{Mohammad Mohammadi Amiri,~\IEEEmembership{Student Member,~IEEE}, and\thanks{M. Mohammadi Amiri was with the Department of Electrical and Electronic Engineering, Imperial College London. He is now with the Department of Electrical Engineering, Princeton University, Princeton, NJ 08544, USA (e-mail: mamiri@princeton.edu).}\thanks{D. G\"und\"uz is with the Department of Electrical and Electronic Engineering, Imperial College London, London SW7 2AZ, U.K. (e-mail: d.gunduz@imperial.ac.uk).}\thanks{This work has been supported by the European Research Council (ERC) through Starting Grant BEACON (agreement No. 677854).}
Deniz G\"und\"uz,~\IEEEmembership{Senior Member,~IEEE}}
}
\date{}

\begin{document}

\maketitle


\begin{abstract}
We study collaborative/ federated machine learning (ML) at the wireless edge, where power and bandwidth-limited wireless devices with local datasets carry out distributed stochastic gradient descent (DSGD) with the help of a remote parameter server (PS). Standard approaches assume separate computation and communication, where local gradient estimates are compressed and transmitted to the PS over orthogonal links. Following this \textit{digital} approach, we introduce D-DSGD, in which the wireless devices employ gradient quantization and error accumulation, and transmit their gradient estimates to the PS over the underlying multiple access channel (MAC).
We then introduce a novel \textit{analog} scheme, called A-DSGD, which exploits the additive nature of the wireless MAC for \textit{over-the-air} gradient computation, and provide convergence analysis for this approach. In A-DSGD, the devices first sparsify their gradient estimates, and then project them to a lower dimensional space imposed by the available channel bandwidth. These projections are sent directly over the MAC without employing any digital code. 
Numerical results show that A-DSGD converges faster than D-DSGD thanks to its more efficient use of the limited bandwidth and the natural alignment of the gradient estimates over the channel. The improvement is particularly compelling at low power and low bandwidth regimes. We also illustrate for a classification problem that, A-DSGD is more robust to bias in data distribution across devices, while D-DSGD significantly outperforms other digital schemes in the literature. We also observe that both D-DSGD and A-DSGD perform better by increasing the number of devices (while keeping the total dataset size constant), showing their ability in harnessing the computation power of edge devices. The lack of quantization and channel encoding/decoding in A-DSGD further speeds up communication, making it very attractive for low-latency ML applications at the wireless edge.
\end{abstract}


\section{Introduction}\label{SecIntro}

Many emerging technologies involve massive amounts of data collection, and collaborative intelligence that can process and make sense of this data. Internet of things (IoT), autonomous driving, or extended reality technologies are prime examples, where data from sensors must be continuously collected, communicated, and processed to make inferences about the state of a system, or predictions about its future states. Many specialized machine learning (ML) algorithms are being developed tailored for various types of sensor data; however, the current trend focuses on centralized algorithms, where a powerful learning algorithm, often a neural network, is trained on a massive dataset. While this inherently assumes the availability of data at a central processor, in the case of wireless edge devices, transmitting the collected data to a central processor in a reliable manner may be too costly in terms of energy and bandwidth, may introduce too much delay, and may infringe users' constraints. Therefore, a much more desirable and practically viable alternative is to develop collaborative ML techniques, also known as \textit{federated learning (FL)}, that can exploit the local datasets and processing capabilities of edge devices, requiring limited communications (see \cite{Park:PIEEE:2018} for a survey of applications of edge intelligence and existing approaches to enable it). In this paper, we consider FL at the wireless network edge, where devices with local datasets with the help of a central parameter server (PS) collaboratively train a learning model.  

ML problems often require the minimization of the empirical loss function
\begin{align}\label{LossFunctionDef}
F \left( \boldsymbol{\theta} \right) = \frac{1}{\left| \mathcal{B} \right|} \sum\nolimits_{ \boldsymbol{u} \in \mathcal{B}} f \left(\boldsymbol{\theta}, \boldsymbol{u} \right),   
\end{align}
where $\boldsymbol{\theta} \in \mathbb{R}^d$ denotes the model parameters to be optimized, $\mathcal{B}$ represents the dataset with size $\left| \mathcal{B} \right|$, and $f(\cdot)$ is the loss function defined by the learning model. The minimization of \eqref{LossFunctionDef} is typically carried out through iterative gradient descent (GD), in which the model parameters at the $t$-th iteration, $\boldsymbol{\theta}_t$, are updated according to 
\begin{align}\label{BasicGDModelUpdate}
\boldsymbol{\theta}_{t+1} =\boldsymbol{\theta}_{t} - \eta_t \nabla F \left( \boldsymbol{\theta}_t \right) = \boldsymbol{\theta}_{t} - \eta_t \frac{1}{\left| \mathcal{B} \right|} \sum\limits_{\boldsymbol{u} \in \mathcal{B}} \nabla f \left(\boldsymbol{\theta}_t, \boldsymbol{u} \right),   
\end{align}
where $\eta_t$ is the learning rate at iteration $t$. However, in the case of massive datasets each iteration of GD becomes prohibitively demanding. Instead, in stochastic GD (SGD) the parameter vector is updated with a stochastic gradient
\begin{align}\label{BasicSGDModelUpdate}
\boldsymbol{\theta}_{t+1} =\boldsymbol{\theta}_{t} - \eta_t \cdot \boldsymbol{g} \left(\boldsymbol{\theta}_{t} \right),   
\end{align}
which satisfies $\mathbb{E} \left[ \boldsymbol{g} \left(\boldsymbol{\theta}_{t} \right) \right] = \nabla F \left( \boldsymbol{\theta}_t \right)$. SGD also allows parallelization when the dataset is distributed across tens or even hundreds of devices. In distributed SGD (DSGD), devices process data samples in parallel while maintaining a globally consistent parameter vector $\boldsymbol{\theta}_{t}$. In each iteration, device $m$ computes a gradient vector based on the global parameter vector with respect to its local dataset $\mathcal{B}_{m}$, and sends the result to the PS. Once the PS receives the computed gradients from all the devices, it updates the global parameter vector according to
\begin{align}\label{ParallelSGDModelUpdate}
\boldsymbol{\theta}_{t+1} =\boldsymbol{\theta}_{t} - \eta_t \frac{1}{M} \sum\nolimits_{m=1}^{M} \boldsymbol{g}_m \left(\boldsymbol{\theta}_{t} \right),   
\end{align}
where $M$ denotes the number of devices, and $\boldsymbol{g}_m \left(\boldsymbol{\theta}_{t} \right) \triangleq \frac{1}{\left| \mathcal{B}_{m} \right|} \sum\nolimits_{\boldsymbol{u} \in \mathcal{B}_{m}} \nabla f \left(\boldsymbol{\theta}_t, \boldsymbol{u} \right)$ is the stochastic gradient of the current model computed at device $m$, $m \in [M]$, using the locally available portion of the dataset, $\mathcal{B}_{m}$. 

While the communication bottleneck in DSGD has been widely acknowledged in the literature \cite{QSGDQuantDAlistarh,DCLinHanDeepGradComp,DCOneBitQuan}, ML literature ignores the characteristics of the communication channel, and simply focus on reducing the amount of data transmitted from each device to the PS. In this paper, we consider DSGD over a shared wireless medium, and explicitly model the channel from the devices to the PS, and treat each iteration of the DSGD algorithm as a \textit{distributed wireless computation problem}, taking into account the physical properties of the wireless medium. 
We will provide two distinct approaches for this problem, based on \textit{digital} and \textit{analog} communications, respectively. We will show that analog ``over-the-air'' computation can significantly speed up wireless DSGD, particularly in bandwidth-limited and low-power settings, typically experienced by wireless edge devices.

\subsection{Prior Works}\label{SecPriorWork}

Extensive efforts have been made in recent years to enable collaborative learning across distributed mobile devices. In the FL framework \cite{FLStrategiesMcMahan,McMahan2017CommunicationEfficientLO,GoogleMcMahanFed}, devices collaboratively learn a model with the help of a PS. In practical implementations, bandwidth of the communication channel from the devices to the PS is the main bottleneck in FL \cite{FLStrategiesMcMahan,McMahan2017CommunicationEfficientLO,GoogleMcMahanFed,FLMultiTaskSmith,FLRandomMeanEst,FLDistOptDecenKonecny,FLClientSelNishio,FLUltReliabLowLatBennis}. Therefore, it is essential to reduce the communication requirements of collaborative ML. To this end, three main approaches, namely \textit{quantization}, \textit{sparsification}, and \textit{local updates}, and their various combinations have been considered in the literature. Quantization methods implement lossy compression of the gradient vectors by quantizing each of their entries to a finite-bit low precision value  \cite{DCOneBitQuan,DCWenTernGrad,DorefaZhou,DCWangATOMO,SignSGDBernstein,ErrorCompSGDWu,HardwareLi,QSGDQuantDAlistarh}. Sparsification reduces the communication time by sending only some values of the gradient vectors \cite{ScalableDNNStorm,DCAjiSparse,DCLinHanDeepGradComp,DCSattlerSparseBinary,SparCMLRenggli,SparseConvergenceAlistarh,VarBasedTsuzuku}, and can be considered as another way of lossy compression, but it is assumed that the chosen entries of the gradient vectors are transmitted reliably, e.g., at a very high resolution. Another approach is to reduce the frequency of communication from the devices by allowing local parameter updates \cite{LocalSGDStich,UseLocalSGDLin,FLStrategiesMcMahan,McMahan2017CommunicationEfficientLO,GoogleMcMahanFed,FLMultiTaskSmith,FLRandomMeanEst,FLDistOptDecenKonecny,FLClientSelNishio,FLUltReliabLowLatBennis}. Although there are many papers in the FL literature, most of these works ignore the physical and network layer aspects of the problem, which are essential for FL implementation at the network edge, where the training takes place over a shared wireless medium. \textit{To the best of our knowledge, this is the first paper to address the bandwidth and power limitations in collaborative edge learning by taking into account the physical properties of the wireless medium.}

\subsection{Our Contributions}\label{SecCont}
Most of the current literature on distributed ML and FL ignore the physical layer aspects of communication, and consider interference-and-error-free links from the devices to the PS to communicate their local gradient estimates, possibly in compressed form to reduce the amount of information that must be transmitted. Due to the prevalence of wireless networks and the increasing availability of edge devices, e.g., mobile phones, sensors, etc., for large-scale data collection and learning, we consider a wireless multiple access channel (MAC) from the devices to the PS, through which the PS receives the gradients computed by the devices at each iteration of the DSGD algorithm. The standard approach to this problem, aligned with the literature on FL, is a \textit{separate} approach to computation and communication.


We first follow this separation-based approach, and propose a digital DSGD scheme, which will be called D-DSGD. In this scheme, the devices first compress their gradient estimates to a finite number of bits through \textit{quantization}. Then, some access scheme, e.g., time, frequency or code division multiple access, combined with error correction coding is employed to transmit compressed gradient estimates to the PS in a reliable manner. In this work, to understand the performance limits of the digital approach, we assume capacity-achieving channel codes are utilized at the devices. The optimal solution for this scheme will require carefully allocating channel resources across the devices and the available power of each device across iterations, together with an efficient gradient quantization scheme. For gradient compression, we will consider state-of-the-art quantization approaches together with local error accumulation \cite{DCSattlerSparseBinary}.

It is known that separation-based approaches to distributed compression and computing over wireless channels are sub-optimal in general \cite{Gunduz:IT:09, DCOverTheAirCompGoldenbaum}. We propose an alternative analog communication approach, in which the devices transmit their local gradient estimates directly over the wireless channel, extending our previous work in \cite{MohammadDenizISIT19}. This scheme is motivated by the fact that the PS is not interested in the individual gradient vectors, but in their average, and the wireless MAC automatically provides the PS with the sum of the gradients (plus a noise term). However, the bandwidth available at each iteration may not be sufficient to transmit the whole gradient vector in an uncoded fashion. Hence, to compress the gradients to the dimension of the limited bandwidth resources, we employ an analog compression scheme inspired by compressive sensing, which was previously introduced in  \cite{SparseCastTungDeniz} for analog image transmission over bandwidth-limited wireless channels. In this analog computation scheme, called A-DSGD, devices first sparsify their local gradient estimates (after adding the accumulated local error). These sparsified gradient vectors are projected to the channel bandwidth using a pseudo-random measurement matrix, as in compressive sensing. Then, all the devices transmit the resultant real-valued vectors to the PS simultaneously in an analog fashion, by simply scaling them to meet the average transmit power constraints. The PS tries to reconstruct the sum of the actual sparse gradient vectors from its noisy observation. We use approximate message passing (AMP) algorithm to do this at the PS \cite{AMPJournalDonoho}. We also prove the convergence of the A-DSGD algorithm considering a strongly convex loss function.

Numerical results show that the proposed analog scheme A-DSGD has a better convergence behaviour compared to its digital counterpart, while D-DSGD outperforms the digital schemes studied in \cite{SignSGDBernstein,QSGDQuantDAlistarh}. Moreover, we observe that the performance of A-DSGD degrades only slightly by introducing bias in the datasets across the devices; while the degradation in D-DSGD is larger, it still outperforms other digital approaches \cite{SignSGDBernstein,QSGDQuantDAlistarh} by a large margin. The performances of both A-DSGD and D-DSGD algorithms improve with the number of devices, keeping the total size of the dataset constant, despite the inherent resource sharing approach for D-DSGD. Furthermore, the performance of A-DSGD degrades only slightly even with a significant reduction in the available average transmit power. We argue that, these benefits of A-DSGD are due to the inherent ability of analog communications to benefit from the signal-superposition characteristic of the wireless medium. Also, reduction in the communication bandwidth of the MAC deteriorates the performance of D-DSGD much more compared to A-DSGD. 

A similar over-the-air computation approach is also considered in parallel works \cite{KaibinParallelWork, YangFedLearOverAirComp}. These works consider, respectively, SISO and SIMO fading MACs from the devices to the PS, and focus on aligning the gradient estimates received from different devices to have the same power at the PS to allow correct computation by performing power control and device selection. Our work is extended to a fading MAC model in \cite{FLMohammadDenizTWCFading}. The distinctive contributions of our work with respect to \cite{KaibinParallelWork} and \cite{YangFedLearOverAirComp} are i) the introduction of a purely digital separate gradient compression and communication scheme; ii) the consideration of a bandwidth-constrained channel, which requires (digital or analog) compression of the gradient estimates; iii) error accumulation at the devices to improve the quality of gradient estimates by keeping track of the information lost due to compression; and iv) power allocation across iterations to dynamically adapt to the diminishing gradient variance. We also remark that both \cite{KaibinParallelWork} and \cite{YangFedLearOverAirComp} consider transmitting model updates, while we focus on gradient transmission, which is more energy-efficient as each device transmits only the innovation obtained through gradient descent at that particular iteration (together with error accumulation), whereas model transmission wastes a significant portion of the transmit power by sending the already known previous model parameters from all the devices. Our work can easily be combined with the federated averaging algorithm in \cite{McMahan2017CommunicationEfficientLO}, where multiple SGD iterations are carried out locally before forwarding the models differences to the PS. We can also apply momentum correction \cite{DCLinHanDeepGradComp}, which improves the convergence speed of the DSGD algorithms with communication delay.


\subsection{Notations}\label{SecNot}
$\mathbb{R}$ represents the set of real values. For positive integer $i$, we let $[i] \triangleq \{ 1, \dots, i \}$, and $\boldsymbol{1}_{i}$ denotes a column vector of length $i$ with all entries $1$. $\mathcal{N} \left( 0,\sigma^2 \right)$ denotes a zero-mean normal distribution with variance $\sigma^2$. We denote the cardinality of set $\mathcal{A}$ by $\left| \mathcal{A} \right|$, and $l_2$ norm of vector $\boldsymbol{x}$ by $\left\| \boldsymbol{x} \right\|$. Also, we represent the $l_2$ induced norm of a rectangular matrix $\boldsymbol{A}$ by $\left\| \boldsymbol{A} \right\|$.


\begin{figure}[!t]
\centering
\includegraphics[scale=0.45]{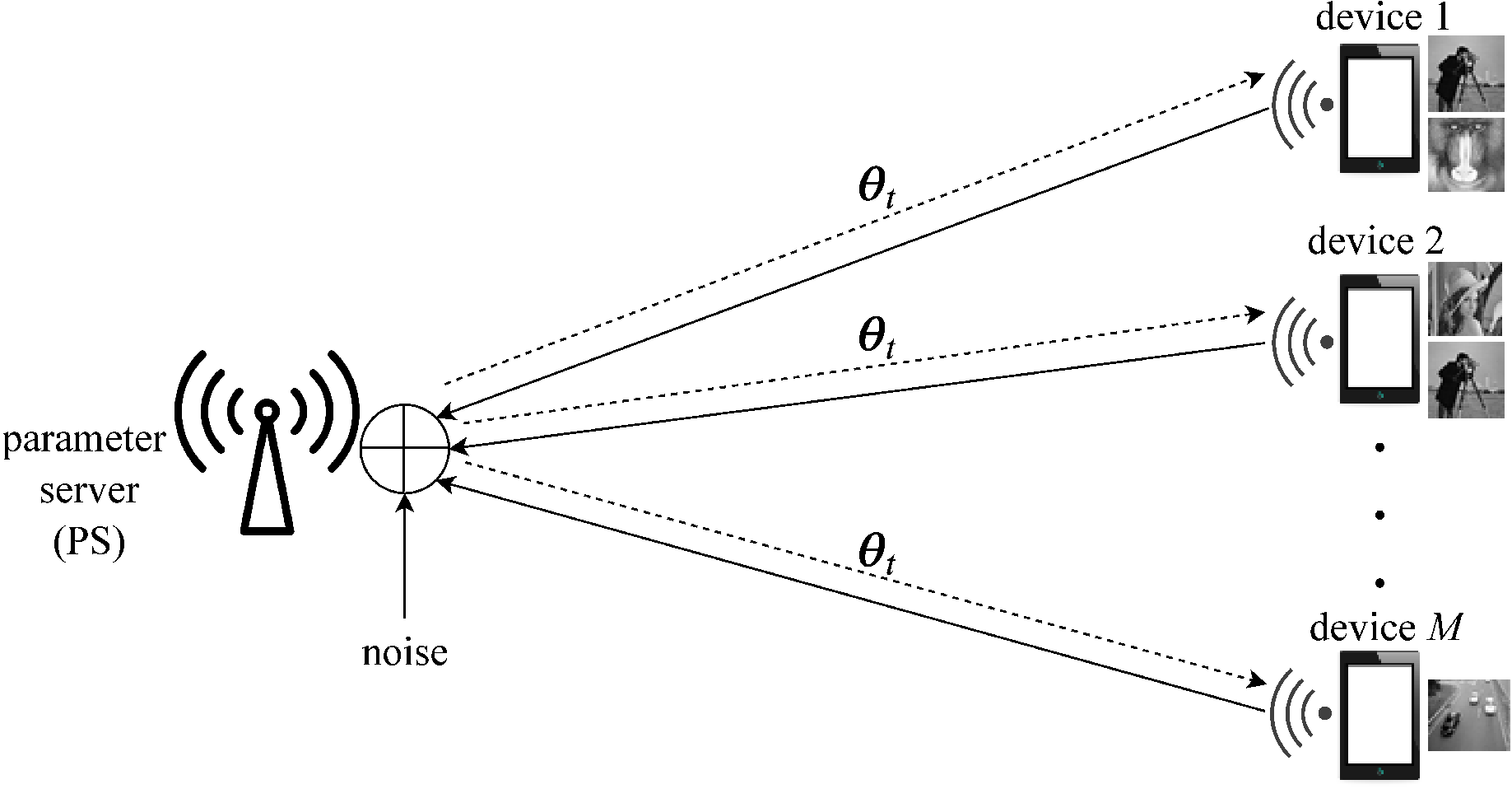}
\caption{Illustration of the FL framework at the wireless edge.} 
\label{System_Model}
\end{figure}

\section{System Model}\label{SecProbFormul}

We consider federated SGD at the wireless edge, where $M$ wireless edge devices employ SGD with the help of a remote PS, to which they are connected through a noisy wireless MAC (see Fig. \ref{System_Model}). Let $\mathcal{B}_{m}$ denote the set of data samples available at device $m$, $m \in [M]$, and $\boldsymbol{g}_m \left(\boldsymbol{\theta}_{t} \right) \in \mathbb{R}^d$ be the stochastic gradient computed by device $m$ using local data samples. At each iteration of the DSGD algorithm in \eqref{ParallelSGDModelUpdate}, the local gradient estimates of the devices are sent to the PS over $s$ uses of a Gaussian MAC, characterized by:
\begin{align}\label{ReceivedVectorPSGen}
\boldsymbol{y}(t) = \sum\nolimits_{m=1}^{M}  \boldsymbol{x}_{m}(t) + \boldsymbol{z}(t),
\end{align}
where $\boldsymbol{x}_{m} (t) \in \mathbb{R}^s$ is the length-$s$ channel input vector transmitted by device $m$ at iteration $t$, $\boldsymbol{y} (t) \in \mathbb{R}^s$ is the channel output received by the PS, and $\boldsymbol{z}(t) \in \mathbb{R}^s$ is the independent additive white Gaussian noise (AWGN) vector with each entry independent and identically distributed (i.i.d.) according to $\mathcal{N} \left( 0, \sigma^2 \right)$. Since we focus on DSGD, the channel input vector of device $m$ at iteration $t$ is a function of the current parameter vector $\boldsymbol{\theta}_{t}$ and the local dataset $\mathcal{B}_{m}$, and more specifically the current gradient estimate at device $m$, $\boldsymbol{g}_m \left(\boldsymbol{\theta}_{t} \right)$, $ m \in [M]$. For a total of $T$ iterations of DSGD algorithm, the following average transmit power constraint is imposed per eache device:
\begin{align}\label{AvePowerConsGen}
\frac{1}{T} \sum\nolimits_{t=1}^{T} ||\boldsymbol{x}_{m} (t)||^2_2 \le \bar{P}, \quad \forall m \in [M],
\end{align}
averaged over iterations of the DSGD algorithm. The goal is to recover the average of the local gradient estimates $\frac{1}{M} \sum\nolimits_{m=1}^{M} \boldsymbol{g}_m \left(\boldsymbol{\theta}_{t} \right)$ at the PS, and update the model parameter as in \eqref{ParallelSGDModelUpdate}. However, due to the pre-processing performed at each device and the noise added by the wireless channel, it is not possible to recover the average gradient perfectly at the PS, and instead, it uses a noisy estimate to update the model parameter vector; i.e., we have $\boldsymbol{\theta}_{t+1} = \phi (\boldsymbol{\theta}_{t}, \boldsymbol{y}(t))$ for some update function $\phi : \mathbb{R}^d \times \mathbb{R}^s \to \mathbb{R}^d$. The updated model parameter is then multicast to the devices by the PS through an error-free shared link. We assume that the PS is not limited in power or bandwidth, so the devices receive a consistent global parameter vector for their computations in the next iteration. 

The transmission of the local gradient estimates, $\boldsymbol{g}_m \left(\boldsymbol{\theta}_{t} \right)$, $\forall m \in [M]$, to the PS with the goal of PS reconstructing their average can be considered as a distributed function computation problem over a MAC \cite{DCOverTheAirCompGoldenbaum}. We will consider both a digital approach treating computation and communication separately, and an analog approach that does not use any coding, and instead applies gradient sparsification followed by a linear transformation to compress the gradients, which are then transmitted simultaneously in an uncoded fashion.

We remark that we assume a simple channel model (Gaussian MAC) for the bandwidth-limited uplink transmission from the devices to the PS in this paper in order to highlight the main ideas behind the proposed digital and analog collaborative learning. The digital and analog approaches proposed in this paper can be extended to more complicated channel models as it has been done in the follow up works \cite{FLMohammadDenizSPAWC19,FLMohammadDenizTWCFading,FLOsvaldoDistillation,FLMohammadTolgaDenizSIPBlind}.

\section{Digital DSGD (D-DSGD)}\label{SecPeoposedDigital}


In this section, we present DSGD at the wireless network edge utilizing digital compression and transmission over the wireless MAC, referred to as the digital DSGD (D-DSGD) algorithm. Since we do not know the variances of the gradient estimates at different devices, we allocate the power equally among the devices, so that device $m$ sends $\boldsymbol{x}_{m} (t)$ with power $P_t$, i.e., $||\boldsymbol{x}_{m} (t)||^2_2 = P_t$, where $P_t$ values are chosen to satisfy the average transmit power constraint over $T$ iterations
\begin{align}\label{AvePowerConsD-DSGDAlg}
\frac{1}{T} \sum\nolimits_{t=1}^{T} P_t \le \bar{P}. 
\end{align}
Due to the intrinsic symmetry of the model, we assume that the devices transmit at the same rate at each iteration (while the rate may change across iterations depending on the allocated power, $P_t$). Accordingly, the total number of bits that can be transmitted from each of the devices over $s$ uses of the Gaussian MAC, described in \eqref{ReceivedVectorPSGen}, is upper bounded by  
\begin{align}\label{SumCapacityWMAC}
R_t = \frac{s}{2M} \log_2 \left( 1 + \frac{M P_t}{s \sigma^2} \right),
\end{align}
where $M P_t / s$ is the sum-power per channel use. Note that this is an upper bound since it is the Shannon capacity of the underlying Gaussian MAC, and we further assumed that the capacity can be achieved over a finite blocklength of $s$.

\begin{remark}\label{RemDifNumBitsDig}
We note that having a distinct sum power $M P_t$ at each iteration $t$ enables each user to transmit different numbers of bits at different iterations. This corresponds to a novel gradient compression scheme for DSGD, in which the devices can adjust over time the amount of information they send to the PS about their gradient estimates. They can send more information bits at the beginning of the DSGD algorithm when the gradient estimates have higher variances, and reduce the number of transmitted bits over time as the variance decreases. We observed empirically that this improves the performance compared to the standard approach in the literature, where the same compression scheme is applied at each iteration \cite{DCSattlerSparseBinary}. 
\end{remark}    


We will adopt the scheme proposed in \cite{DCSattlerSparseBinary} for gradient compression at each iteration of the DSGD scheme.
We modify this scheme by allowing different numbers of bits to be transmitted by the devices at each iteration. 
At each iteration the devices sparsify their gradient estimates as described below. In order to retain the accuracy of their local gradient estimates, devices employ \textit{error accumulation} \cite{DCOneBitQuan,Strom2015ScalableDD}, where the accumulated error vector at device $m$ until iteration $t$ is denoted by ${\boldsymbol{\Delta}_{m} (t)} \in \mathbb{R}^d$, where we set ${\boldsymbol{\Delta}_{m} (0)} = \boldsymbol{0}$, $\forall m \in [M]$. Hence, after the computation of the local gradient estimate for parameter vector $\boldsymbol{\theta}_t$, i.e., $\boldsymbol{g}_m \left(\boldsymbol{\theta}_t \right)$, device $m$ updates its estimate with the accumulated error as $\boldsymbol{g}_m \left(\boldsymbol{\theta}_t \right) + {\boldsymbol{\Delta}_{m} (t)}$, $m \in [M]$. At iteration $t$, device $m$, $m \in [M]$, sets all but the highest $q_t$ and the smallest $q_t$ of the entries of its gradient estimate vector $\boldsymbol{g}_m \left(\boldsymbol{\theta}_t \right) + {\boldsymbol{\Delta}_{m} (t)}$, of dimension $d$, to zero, where $q_t \le d/2$ (to have a communication-efficient scheme, in practice, the goal is to have $q_t \ll  d$, $\forall t$). Then, it computes the mean values of all the remaining positive entries and all the remaining negative entries, denoted by $\mu^{+}_{m} (t)$ and $\mu^{-}_{m} (t)$, respectively. If $\mu^{+}_{m} (t) > \left| \mu^{-}_{m} (t) \right|$, then it sets all the entries with negative values to zero and all the entries with positive values to $\mu^{+}_{m} (t)$, and vice versa. We denote the resulting sparse vector at device $m$ by ${\boldsymbol{g}}^q_{m} \left( \boldsymbol{\theta}_t \right)$, and device $m$ updates the local accumulated error vector as ${\boldsymbol{\Delta}_{m} (t+1)} = \boldsymbol{g}_m \left(\boldsymbol{\theta}_t \right) + {\boldsymbol{\Delta}_{m} (t)} - {\boldsymbol{g}}^q_{m} \left( \boldsymbol{\theta}_t \right)$, $m \in [M]$. It then aims to send ${\boldsymbol{g}}^q_{m} \left( \boldsymbol{\theta}_t \right)$ over the channel by transmitting its mean value and the positions of its non-zero entries. For this purpose, we use a 32-bit representation of the absolute value of the mean (either $\mu^{+}_{m} (t)$ or $\left| \mu^{-}_{m} (t) \right|$) along with 1 bit indicating its sign. To send the positions of the non-zero entries, it is assumed in \cite{DCSattlerSparseBinary} that the distribution of the distances between the non-zero entries is geometrical with success probability $q_t$, which allows them to use Golomb encoding to send these distances with a total number of bits $b^* + \frac{1}{1-(1-q_t)^{2^{b^*}}}$,
where $b^* = 1 + \left\lfloor {{{\log }_2}\left( {\frac{{\log \left( {(\sqrt 5  - 1)/2} \right)}}{{\log \left( {1 - {q_t}} \right)}}} \right)} \right\rfloor$. However, we argue that, sending $\log_2 \binom{d}{q_t}$ bits to transmit the positions of the non-zero entries is sufficient regardless of the distribution of the positions. This can be achieved by simply enumerating all possible sparsity patterns. Thus, with D-DSGD, the total number of bits sent by each device at iteration $t$ is given by 
\begin{align}\label{SWMSNumberBits}
r_t = \log_2 \binom{d}{q_t} + 33,
\end{align}
where $q_t$ is chosen as the highest integer satisfying $r_t \le R_t$.

We will also study the impact of introducing more devices into the system. With the reducing cost of sensing and computing devices, we can consider introducing more devices, each coming with its own power resources. We assume that the size of the total dataset remains constant, which allows each sensor to save computation time and energy. Assume that the number of devices is increased by a factor $\kappa > 1$. We assume that the total power consumed by the devices at iteration $t$ also increases by factor $\kappa$. We can see that the maximum number of bits that can be sent by each device is strictly smaller in the new system. this means that the PS receives a less accurate gradient estimate from each device, but from more devices. 
Numerical results for the D-DSGD scheme and its comparison to analog transmission are relegated to Section \ref{SecExperiments}.

\begin{remark}\label{RemDigital_SDGD_QDSGD}
We have established a framework for digital transmission of the gradients over a band-limited wireless channel, where other gradient quantization techniques can also be utilized. In Section \ref{SecExperiments}, we will compare the performance of the proposed D-DSGD scheme with that of sign DSGD (S-DSGD) \cite{SignSGDBernstein} and quantized DSGD (Q-DSGD) \cite{QSGDQuantDAlistarh}, adopted for digital transmission over a limited capacity MAC.   
\end{remark}


\section{Analog DSGD (A-DSGD)}\label{SecPeoposedAnalog}

Next, we propose an analog DSGD algorithm, called A-DSGD, which does not employ any digital code, either for compression or channel coding, and instead all the devices transmit their gradient estimates simultaneously in an uncoded manner. This is motivated by the fact that the PS is not interested in the individual gradients, but only in their average. The underlying wireless MAC provides the sum of the gradients, which is used by the PS to update the parameter vector. See Algorithm \ref{A-DSGD_alg} for a description of the A-DSGD scheme. 

Similarly to D-DSGD, devices employ local error accumulation. Hence, after computing the local gradient estimate for parameter vector $\boldsymbol{\theta}_t$, each device updates its estimate with the accumulated error as $\boldsymbol{g}_m^{ec} \left(\boldsymbol{\theta}_t \right) \triangleq \boldsymbol{g}_m \left(\boldsymbol{\theta}_t \right) + {\boldsymbol{\Delta}_{m}(t)}$, $m \in [M]$.


The challenge in the analog transmission approach is to compress the gradient vectors to the available channel bandwidth. In many modern ML applications, such as deep neural networks, the parameter vector, and hence the gradient vectors, have extremely large dimensions, whereas the channel bandwidth, measured by parameter $s$, is small due to the bandwidth and latency limitations. Thus, transmitting all the model parameters one-by-one in an uncoded/analog fashion is not possible as we typically have $d \gg s$. 

Lossy compression at any required level is at least theoretically possible in the digital domain. For the analog scheme, in order to reduce the dimension of the gradient vector to that of the channel, the devices apply gradient sparsification. In particular, device $m$ sets all but the $k$ elements of the error-compensated resulting vector $\boldsymbol{g}_m^{ec} \left(\boldsymbol{\theta}_t \right)$ with the highest magnitudes to zero. We denote the sparse vector at device $m$ by $\boldsymbol{g}_m^{sp} \left(\boldsymbol{\theta}_t \right)$, $m \in [M]$. This $k$-level sparsification is represented by function ${\rm{sp}}_k$ in Algorithm \ref{ModelUpdateAlg}, i.e., $\boldsymbol{g}_m^{sp} \left(\boldsymbol{\theta}_t \right) = {\rm{sp}}_k \left( \boldsymbol{g}_m^{ec} \left(\boldsymbol{\theta}_t \right) \right)$. The accumulated error at device $m$, $m \in [M]$, which is the difference between $\boldsymbol{g}_m^{ec} \left(\boldsymbol{\theta}_t \right)$ and its sparsified version, i.e., it maintains those elements of vector $\boldsymbol{g}_m^{ec} \left(\boldsymbol{\theta}_t \right)$ which are set to zero as the result of sparsification at iteration $t$, is then updated according to 
\begin{align}\label{AccumErrorUpdateItet}
\boldsymbol{\Delta}_{m} (t+1) =& \boldsymbol{g}_m^{ec} \left(\boldsymbol{\theta}_t \right) - \boldsymbol{g}_m^{sp} \left(\boldsymbol{\theta}_t \right)= \boldsymbol{g}_m \left(\boldsymbol{\theta}_t \right) + \boldsymbol{\Delta}_{m} (t)  \nonumber\\
&\qquad \qquad \quad - {\rm{sp}}_k \left( \boldsymbol{g}_m \left(\boldsymbol{\theta}_t \right) + \boldsymbol{\Delta}_{m} (t) \right).
\end{align}

We would like to transmit only the non-zero entries of these sparse vectors. However, simply ignoring the zero elements would require transmitting their indeces to the PS separately. To avoid this additional data transmission, we will employ a random projection matrix, similarly to compressive sensing. A similar idea is recently used in \cite{SparseCastTungDeniz} for analog image transmission over a bandwidth-limited channel. 

\begin{algorithm}[t]
\caption{A-DSGD}
\label{ModelUpdateAlg}
\begin{algorithmic}[1]
\Statex
\State{\textbf{Initialize} $\boldsymbol{\theta}_0 = \boldsymbol{0}$ and $\boldsymbol{\Delta}_{1} (0) = \cdots = \boldsymbol{\Delta}_{M} (0) = \boldsymbol{0}$}
\For {$t = 0, \ldots, T-1$}
\Statex
\begin{itemize}
\item \textbf{devices do:}
\end{itemize}
\For {$m = 1, \ldots, M$ in parallel}
\State{Compute $\boldsymbol{g}_m \left( \boldsymbol{\theta}_t \right)$ with respect to $\mathcal{B}_{m}$}
\State{$\boldsymbol{g}_m^{ec} \left( \boldsymbol{\theta}_t \right) = \boldsymbol{g}_m \left( \boldsymbol{\theta}_t \right) + {\boldsymbol{\Delta}_{m} (t)}$}
\State{$\boldsymbol{g}_m^{sp} \left( \boldsymbol{\theta}_t \right) = {\rm{sp}}_k \left( \boldsymbol{g}_m^{ec} \left( \boldsymbol{\theta}_t \right) \right)$}
\State{$\boldsymbol{\Delta}_{m} (t+1) = \boldsymbol{g}_m^{ec} \left( \boldsymbol{\theta}_t \right) - \boldsymbol{g}_m^{sp} \left( \boldsymbol{\theta}_t \right)$}
\State{$\tilde{\boldsymbol{g}}_m \left( \boldsymbol{\theta}_{t} \right) = \boldsymbol{A}_{s-1} \boldsymbol{g}_m^{sp} \left( \boldsymbol{\theta}_t \right)$}
\State{$\boldsymbol{x}_{m} \left( t \right) = \begin{bmatrix}
    {\sqrt{\alpha_{m}(t)} \tilde{\boldsymbol{g}}_m \left(\boldsymbol{\theta}_t \right)}^T &
    \sqrt{\alpha_{m}(t)} 
\end{bmatrix}^T$}
\EndFor
\Statex
\begin{itemize}
\item \textbf{PS does:}
\end{itemize}
\State{\quad \; $\hat{\boldsymbol{g}} \left(\boldsymbol{\theta}_{t} \right) = {\rm{AMP}}_{\boldsymbol{A}_{s-1}} \left( \frac{1}{y_{s} \left( t \right)} \boldsymbol{y}^{s-1} \left( t \right) \right)$}
\State{\quad \; $\boldsymbol{\theta}_{t+1} =\boldsymbol{\theta}_{t} - \eta_t \cdot \hat{\boldsymbol{g}} \left(\boldsymbol{\theta}_{t} \right)$}
\EndFor
\end{algorithmic}\label{A-DSGD_alg}
\end{algorithm}

A pseudo-random matrix $\boldsymbol{A}_{\tilde{s}} \in \mathbb{R}^{\tilde{s} \times d}$, for some $\tilde{s} \le s$, with each entry i.i.d. according to $\mathcal{N} (0,1/\tilde{s})$, is generated and shared between the PS and the devices before starting the computations. At iteration $t$, device $m$ computes $\tilde{\boldsymbol{g}}_m \left(\boldsymbol{\theta}_{t} \right) \triangleq \boldsymbol{A}_{\tilde{s}} \boldsymbol{g}_m^{sp} \left(\boldsymbol{\theta}_t \right) \in \mathbb{R}^{\tilde{s}}$, and transmits $\boldsymbol{x}_{m}(t) = \begin{bmatrix}
    {\sqrt{\alpha_{m}(t)} \tilde{\boldsymbol{g}}_m \left(\boldsymbol{\theta}_{t} \right)}^T & {\boldsymbol{a}_{m} (t)}^T
\end{bmatrix}^T$, where ${\boldsymbol{a}_{m} (t)} \in \mathbb{R}^{s - \tilde{s}}$, over the MAC, $m \in [M]$, while satisfying the average power constraint \eqref{AvePowerConsGen}. The PS receives 
\begin{align}\label{ReceivedVectorPSDSGDCS}
\boldsymbol{y} \left( t \right) & = \sum\nolimits_{m=1}^{M} {\boldsymbol{x}}_{m} \left( t \right) + \boldsymbol{z} (t) \nonumber\\
& = \begin{bmatrix}
    \boldsymbol{A}_{\tilde{s}} \sum\nolimits_{m=1}^{M} \sqrt{\alpha_{m}(t)} \boldsymbol{g}_m^{sp} \left(\boldsymbol{\theta}_t \right) \\
    \sum\nolimits_{m=1}^{M} {\boldsymbol{a}_{m}(t)} 
\end{bmatrix}  
+ \boldsymbol{z} (t).
\end{align}
In the following, we propose a power allocation scheme designing the scaling coefficients $\sqrt{\alpha_m(t)}$, $\forall m,t$.


We set $\tilde{s} = s-1$, which requires $s \ge 2$. At iteration $t$, we set ${\boldsymbol{a}_{m}(t)} = \sqrt{\alpha_{m}(t)}$, and device $m$ computes $\tilde{\boldsymbol{g}}_m \left(\boldsymbol{\theta}_{t} \right) = \boldsymbol{A}_{s-1} \boldsymbol{g}_m^{sp} \left(\boldsymbol{\theta}_t \right)$, and sends vector $\boldsymbol{x}_{m} \left( t \right) = \begin{bmatrix}
    {\sqrt{\alpha_{m}(t)} \tilde{\boldsymbol{g}}_m \left(\boldsymbol{\theta}_t \right)}^T &
    \sqrt{\alpha_{m}(t)} 
\end{bmatrix}^T$ with the same power $P_t = || \boldsymbol{x}_{m} \left( t \right) ||_2^2$ satisfying the average power constraint $\frac{1}{T} \sum\nolimits_{t=1}^{T} P_t \le \bar{P}$, for $m \in [M]$. Accordingly, scaling factor $\sqrt{\alpha_{m}(t)}$ is determined to satisfy
\begin{align}\label{ScaleFacCalcUPA1}
{P}_t = \alpha_{m}(t) \left(  \left\| \tilde{\boldsymbol{g}}_m \left( \boldsymbol{\theta}_t \right) \right\|_2^2 + 1 \right),
\end{align}
which yields
\begin{align}\label{ScaleFacCalcUPA2}
\alpha_{m}(t) = \frac{{P}_t}{\left\| \tilde{\boldsymbol{g}}_m \left( \boldsymbol{\theta}_t \right) \right\|_2^2 + 1}, \quad \mbox{for $m \in [M]$}. 
\end{align}

Since $\left\| \tilde{\boldsymbol{g}}_m \left( \boldsymbol{\theta}_t \right) \right\|_2^2$ may vary across devices, so can $\sqrt{\alpha_{m}(t)}$. That is why, at iteration $t$, device $m$ allocates one channel use to provide the value of $\sqrt{\alpha_{m}(t)}$ to the PS along with its scaled low-dimensional gradient vector $\tilde{\boldsymbol{g}}_m \left( \boldsymbol{\theta}_t \right)$, $m \in [M]$. Accordingly, the received vector at the PS is given by
\begin{align}\label{ReceivedVectorPSUPA}
\boldsymbol{y} \left(\boldsymbol{\theta}_{t} \right) = \begin{bmatrix}
    \boldsymbol{A}_{s-1} \sum\nolimits_{m=1}^{M} \sqrt{\alpha_{m}(t)} \boldsymbol{g}_m^{sp} \left(\boldsymbol{\theta}_t \right) \\
    \sum\nolimits_{m=1}^{M} {\sqrt{\alpha_{m}(t)}} 
\end{bmatrix} + \boldsymbol{z} (t),
\end{align}
where $\alpha_{m}(t)$ is replaced by \eqref{ScaleFacCalcUPA2}. For $i \in [s]$, we define
\begin{align}\label{DefVec1toi}
\boldsymbol{y}^{i} \left( t \right) & \triangleq \begin{bmatrix}
    y_1 \left( t \right) & y_2 \left( t \right) & \cdots &  y_i \left( t \right)
\end{bmatrix}^T,\\
\boldsymbol{z}^{i} (t) & \triangleq \begin{bmatrix}
    z_{1} (t) & z_{2} (t) & \cdots &  z_{i} (t)
\end{bmatrix}^T,
\end{align}
where $y_j \left( t \right)$ and $z_{j}(t)$ denote the $j$-th entry of $\boldsymbol{y} \left( t \right)$ and $\boldsymbol{z} (t)$, respectively. Thus, we have
\begin{subequations}
\label{ReceivedVectorPSUPADecomposed}
\begin{align}\label{ReceivedVectorPSUPADecomposeds_1}
\boldsymbol{y}^{s-1} \left( t \right) & = \boldsymbol{A}_{s-1} \sum\nolimits_{m=1}^{M} \sqrt{\alpha_{m}(t)} \boldsymbol{g}_m^{sp} \left(\boldsymbol{\theta}_t \right) + \boldsymbol{z}^{s-1} (t), \\
y_{s} \left( t \right) &= \sum\nolimits_{m=1}^{M} \sqrt{\alpha_{m}(t)} + z_{s}(t).  
\label{ReceivedVectorPSUPADecomposeds}
\end{align}
\end{subequations}
Note that the goal is to recover $\frac{1}{M}\sum\nolimits_{m=1}^{M} \boldsymbol{g}_m^{sp} \left(\boldsymbol{\theta}_{t} \right)$ at the PS, while, from $\boldsymbol{y}^{s-1} \left( t \right)$ given in \eqref{ReceivedVectorPSUPADecomposeds_1}, the PS observes a noisy version of the weighted sum $\sum\nolimits_{m=1}^{M} \sqrt{\alpha_{m}(t)} \boldsymbol{g}_m^{sp} \left(\boldsymbol{\theta}_t \right)$ projected into a low-dimensional vector through $\boldsymbol{A}_{s-1}$. According to \eqref{ScaleFacCalcUPA2}, each value of $\left\| \tilde{\boldsymbol{g}}_m \left( \boldsymbol{\theta}_t \right) \right\|_2^2$ results in a distinct scaling factor $\alpha_{m}(t)$. However, for large enough $d$ and $\left| \mathcal{B}_{m} \right|$, the values of $\left\| \tilde{\boldsymbol{g}}_m \left( \boldsymbol{\theta}_t \right) \right\|_2^2, \forall m \in [M]$, are not going to be too different across devices. 
As a result, scaling factors $\sqrt{\alpha_{m}(t)}, \forall m \in [M]$, are not going to be very different either. Accordingly, to diminish the effect of scaled gradient vectors, we choose to scale down the received vector $\boldsymbol{y}^{s-1} \left( t \right)$ at the PS, given in \eqref{ReceivedVectorPSUPADecomposeds_1}, with the sum of the scaling factors, i.e., $\sum\nolimits_{m=1}^{M} \sqrt{\alpha_{m} (t)}$, whose noisy version is received by the PS as $y_{s} \left( t \right)$ given in \eqref{ReceivedVectorPSUPADecomposeds}. The resulting scaled vector at the PS is 
\begin{align}\label{ReceivedVectorPSGenDeScaleUPA}
\frac{\boldsymbol{y}^{s-1} \left( t \right)}{y_{s} \left( t \right)} = & \boldsymbol{A}_{s-1} \sum\nolimits_{m=1}^{M} \frac{ \sqrt{\alpha_{m} (t)}}{\sum\nolimits_{i=1}^{M} \sqrt{\alpha_{i} (t)} + z_{s} (t)} \boldsymbol{g}_m^{sp} \left(\boldsymbol{\theta}_t \right) \nonumber\\
& + \frac{1}{\sum\nolimits_{i=1}^{M} \sqrt{\alpha_{i} (t)} + z_{s} (t)} \boldsymbol{z}^{s-1} (t), 
\end{align}
where $\alpha_{m} (t)$, $m \in [M]$, is given in \eqref{ScaleFacCalcUPA2}. By our choice, the PS tries to recover $\frac{1}{M}\sum\nolimits_{m=1}^{M} \boldsymbol{g}_m^{sp} \left(\boldsymbol{\theta}_{t} \right)$ from $\boldsymbol{y}^{s-1} \left( t \right) / {y}_{s} \left( t \right)$ knowing the measurement matrix $\boldsymbol{A}_{s-1}$. The PS estimates $\hat{\boldsymbol{g}} \left(\boldsymbol{\theta}_{t} \right)$ using the AMP algorithm. The estimate $\hat{\boldsymbol{g}} \left(\boldsymbol{\theta}_{t} \right)$ is then used to update the model parameter as 
$\boldsymbol{\theta}_{t+1} =\boldsymbol{\theta}_{t} - \eta_t \cdot \hat{\boldsymbol{g}} \left(\boldsymbol{\theta}_{t} \right)$.


\begin{remark}\label{RemDifEPAandUPA}
We remark here that, with SGD the empirical variance of the stochastic gradient vectors reduce over time approaching zero asymptotically. The power should be allocated over iterations taking into account this decaying behaviour of gradient variance, while making sure that the noise term would not become dominant. To reduce the variation in the scaling factors $\sqrt{\alpha_{m}(t)}, \forall m, t$, variance reduction techniques can be used \cite{AcceleratingVarianceJohnson}. We also note that setting $P_t = \bar{P}$, $\forall t$, results in a special case, where the power is allocated uniformly over time to be resistant against the noise term. 
\end{remark}

\begin{remark}\label{RemAnalogm}
Increasing $M$ can help increase the convergence speed for A-DSGD. This is due to the fact that having more signals superposed over the MAC leads to more robust transmission against noise, particularly when the ratio $\bar{P} / s \sigma^2$ is relatively small, as we will observe in Fig. \ref{ADSGD_DDSGD_M_B}. Also, for a larger $M$ value, $\frac{1}{M} \sum\nolimits_{m=1}^{M} \boldsymbol{g}_m^{sp} \left(\boldsymbol{\theta}_{t} \right)$ provides a better estimate of $\frac{1}{M} \sum\nolimits_{m=1}^{M} \boldsymbol{g}_m \left(\boldsymbol{\theta}_{t} \right)$, and receiving information from a larger number of devices can make these estimates more reliable.
\end{remark}

\begin{remark}\label{RemTradeoffkands}
In the proposed A-DSGD algorithm, the sparsification level, $k$, results in a trade-off. For a relatively small value of $k$, $\frac{1}{M} \sum\nolimits_{m=1}^{M} \boldsymbol{g}_m^{sp} \left(\boldsymbol{\theta}_t \right)$ can be more reliably recovered from $\frac{1}{M} \sum\nolimits_{m=1}^{M} \tilde{\boldsymbol{g}}_m \left( \boldsymbol{\theta}_{t} \right)$; however it may not provide an accurate estimate of the actual average gradient $\frac{1}{M} \sum\nolimits_{m=1}^{M} \boldsymbol{g}_m \left(\boldsymbol{\theta}_t \right)$. Whereas, with a higher $k$ value, $\frac{1}{M} \sum\nolimits_{m=1}^{M} \boldsymbol{g}_m^{sp} \left(\boldsymbol{\theta}_t \right)$ provides a better estimate of $\frac{1}{M} \sum\nolimits_{m=1}^{M} \boldsymbol{g}_m \left(\boldsymbol{\theta}_t \right)$, but reliable recovery of $\frac{1}{M} \sum\nolimits_{m=1}^{M} \boldsymbol{g}_m^{sp} \left(\boldsymbol{\theta}_t \right)$ from the vector $\frac{1}{M} \sum\nolimits_{m=1}^{M} \tilde{\boldsymbol{g}}_m \left( \boldsymbol{\theta}_{t} \right)$ is less likely.     
\end{remark}



\subsection{Mean-Removal for Efficient Transmission}\label{SubSecMean}
To have a more efficient usage of the available power, each device can remove the mean value of its gradient estimate before scaling and sending it. We define the mean value of $\tilde{\boldsymbol{g}}_m \left( \boldsymbol{\theta}_t \right) = \boldsymbol{A}_{\tilde{s}} \boldsymbol{g}_m \left(\boldsymbol{\theta}_t \right)$ as
${\mu}_{m} (t) \triangleq \frac{1}{\tilde{s}} \sum\nolimits_{i=1}^{\tilde{s}} {\tilde{g}}_{m,i} \left( \boldsymbol{\theta}_t \right)$, for $m \in [M]$,
where ${\tilde{g}}_{m,i} \left( \boldsymbol{\theta}_t \right)$ is the $i$-th entry of vector $\tilde{\boldsymbol{g}}_m \left( {\boldsymbol\theta}_t \right)$, $i \in [\tilde{s}]$. We also define $\tilde{\boldsymbol{g}}_m^{az} \left( \boldsymbol{\theta}_t \right) \triangleq \tilde{\boldsymbol{g}}_m \left( \boldsymbol{\theta}_t \right) - {\mu}_{m} (t) \boldsymbol{1}_{\tilde{s}}$, $m \in [M]$. The power of vector $\tilde{\boldsymbol{g}}_m^{az} \left( \boldsymbol{\theta}_t \right)$ is given by
\begin{align}\label{PowerZeroMeanVec}
\left\| \tilde{\boldsymbol{g}}_m^{az} \left( \boldsymbol{\theta}_t \right) \right\|_2^2 = \left\| \tilde{\boldsymbol{g}}_m \left( \boldsymbol{\theta}_t \right) \right\|_2^2 - \tilde{s} {\mu}^2_{m} (t).
\end{align}

We set $\tilde{s} = s-2$, which requires $s \ge 3$. We also set ${\boldsymbol{a}_{m} (t)} = \begin{bmatrix}
    \sqrt{\alpha^{az}_{m} (t)} {\mu}_{m} (t) & \sqrt{\alpha^{az}_{m} (t)}
\end{bmatrix}^T$, and after computing $\tilde{\boldsymbol{g}}_m \left(\boldsymbol{\theta}_{t} \right) = \boldsymbol{A}_{s-2} \boldsymbol{g}_m^{sp} \left(\boldsymbol{\theta}_t \right)$, device $m$, $m \in [M]$, sends 
\begin{align}\label{devicemSendsMR_UPA}
\boldsymbol{x}_{m} \left( t \right) = \begin{bmatrix}
    {\sqrt{\alpha_{m}^{az} (t)} \tilde{\boldsymbol{g}}_m^{az} \left(\boldsymbol{\theta}_t \right)}^T &
    \sqrt{\alpha^{az}_{m} (t)} {\mu}_{m} (t) & \sqrt{\alpha^{az}_{m} (t)}
\end{bmatrix}^T, 
\end{align}
with power
\begin{align}\label{InsPowerUPAZM}
{\small{ || \boldsymbol{x}_{m} \left( t \right) ||_2^2 = \alpha^{az}_{m} (t) \left( \left\| \tilde{\boldsymbol{g}}_m \left( \boldsymbol{\theta}_t \right) \right\|_2^2 - (s-3) {\mu}_{m} (t)^2 + 1 \right),}}
\end{align}
which is chosen to be equal to $P_t$, such that $\frac{1}{T} \sum\nolimits_{t=1}^{T} P_t \le \bar{P}$. Thus, we have, for $m \in [M]$,
\begin{align}\label{ScaleFacCalcUPA2ZM}
\alpha_{m}^{az} (t) = \frac{P_t}{\left\| \tilde{\boldsymbol{g}}_m \left( \boldsymbol{\theta}_t \right) \right\|_2^2 - (s-3) {\mu}_{m}^2 (t) + 1}. 
\end{align}
Compared to the transmit power given in \eqref{ScaleFacCalcUPA1}, we observe that removing the mean reduces the transmit power by $\alpha_{m}^{az} (t) (s-3) {\mu}_{m}^2 (t)$, $m \in [M]$. The received vector at the PS is given by
\begin{align}\label{ReceivedVectorPSUPAZM}
&\boldsymbol{y} \left(\boldsymbol{\theta}_{t} \right) = \begin{bmatrix}
     \sum\nolimits_{m=1}^{M} \sqrt{{\alpha}_{m}^{az} (t)}  \tilde{\boldsymbol{g}}_m^{az} \left( \boldsymbol{\theta}_t \right) \\
     \sum\nolimits_{m=1}^{M} \sqrt{\alpha_{m}^{az} (t)} {\mu}_{m} (t) \\
     \sum\nolimits_{m=1}^{M} \sqrt{\alpha_{m}^{az} (t)}
\end{bmatrix} + \boldsymbol{z} (t) \nonumber\\
& = \begin{bmatrix}
     \sum\nolimits_{m=1}^{M} \sqrt{{\alpha}_{m}^{az} (t)} \left( \boldsymbol{A}_{s-2} {\boldsymbol{g}}_m^{sp} \left( \boldsymbol{\theta}_t \right) - \mu_{m} (t) \boldsymbol{1}_{s-2} \right) \\
    \sum\nolimits_{m=1}^{M} \sqrt{\alpha_{m}^{az} (t)} {\mu}_{m} (t) \\
     \sum\nolimits_{m=1}^{M} \sqrt{\alpha_{m}^{az} (t)}
\end{bmatrix} + \boldsymbol{z} (t),
\end{align}
where we have
\begin{subequations}
\label{ReceivedVectorPSMR_UPADecomposed}
\begin{align}\label{ReceivedVectorPSMR_UPADecomposeds_2}
\boldsymbol{y}^{s-2} \left( t \right) & = \boldsymbol{A}_{s-2} \sum\nolimits_{m=1}^{M} \sqrt{{\alpha}_{m}^{az} (t)} \boldsymbol{g}_m^{sp} \left(\boldsymbol{\theta}_t \right) \nonumber\\
& \;\; \; - \sum\nolimits_{m=1}^{M} \sqrt{{\alpha}_{m}^{az} (t)} \mu_{m} (t) \boldsymbol{1}_{s-2} + \boldsymbol{z}^{s-2} (t), \\
y_{s-1} \left( t \right) &=  \sum\nolimits_{m=1}^{M} \sqrt{\alpha_{m}^{az} (t)} {\mu}_{m} (t) + z_{s-1} (t),  
\label{ReceivedVectorPSMR_UPADecomposeds_1}\\
y_{s} \left( t \right) &= \sum\nolimits_{m=1}^{M} \sqrt{\alpha_{m}^{az} (t)} + z_{s} (t).  
\label{ReceivedVectorPSMR_UPADecomposeds}
\end{align}
\end{subequations}
The PS performs AMP to recover $\frac{1}{M}\sum\nolimits_{m=1}^{M} \boldsymbol{g}_m^{sp} \left(\boldsymbol{\theta}_{t} \right)$ from the following vector
\begin{align}\label{ReceivedVectorPSGenDeScaleUPAZM}
& \frac{1}{y_{s} \left( t \right)} \left( \boldsymbol{y}^{s-2} \left( t \right) + y_{s-1} \left( t \right) \boldsymbol{1}_{s-2} \right) = \nonumber\\
&\qquad \qquad   \boldsymbol{A}_{s-2} \sum\nolimits_{m=1}^{M} \frac{\sqrt{\alpha_{m}^{az} (t)}}{\sum\nolimits_{i=1}^{M} \sqrt{\alpha_{i}^{az} (t)} + z_{s} (t)} \boldsymbol{g}_m^{sp} \left(\boldsymbol{\theta}_t \right) \nonumber\\
&\qquad \qquad   + \frac{z_{s-1} (t) \boldsymbol{1}_{s-2} + \boldsymbol{z}^{s-2} (t)}{\sum\nolimits_{i=1}^{M} \sqrt{\alpha_{i}^{az} (t)} + z_{s} (t)}. 
\end{align}


\section{Convergence Analysis of A-DSGD Algorithm}\label{SecConvergence}
In this section we provide convergence analysis of A-DSGD presented in Algorithm \ref{ModelUpdateAlg}. For simplicity, we assume that $\eta_t = \eta$, $\forall t$. We consider a differentiable and $c$-strongly convex loss function $F$, i.e., $\forall \boldsymbol{x}, \boldsymbol{y} \in \mathbb{R}^d$, it satisfies 
\begin{align}\label{c_StrongConvexFun}
F(\boldsymbol{y}) - F(\boldsymbol{x}) \ge \nabla F \left(\boldsymbol{x} \right)^{T} (\boldsymbol{y} - \boldsymbol{x}) + \frac{c}{2} \left\| \boldsymbol{y} - \boldsymbol{x} \right\|^2.       
\end{align}

\subsection{Preliminaries}\label{ConvergPre}

We first present some preliminaries and background for the essential instruments of the convergence analysis. 

\begin{assumption}\label{AssumpVarianceGrads}
The average of the first moment of the gradients at different devices is bounded as follows:
\begin{align}\label{ThetaBarModelUpdateDevicem}
\frac{1}{M} \sum\nolimits_{m=1}^{M} \mathbb{E}\left[ \left\| {\boldsymbol{g}}_m \left( \boldsymbol{\theta}_t \right) \right\| \right] \le G, \quad \forall t,
\end{align}
where the expectation is over the randomness of the gradients.
\end{assumption}

\begin{assumption}\label{AssumAlpham}
For scaling factors $\sqrt{\alpha_{m}(t)}$, $\forall m \in [M]$, given in \eqref{ScaleFacCalcUPA2}, we approximate
$\frac{ \sqrt{\alpha_{m} (t)}}{\sum\nolimits_{i=1}^{M} \sqrt{\alpha_{i} (t)} + z_{s} (t)}  \approx \frac{1}{M}$, $\forall t$.
\end{assumption}
The rationale behind Assumption \ref{AssumAlpham} is assuming that the noise term is small compared to the sum of the scaling factors, $\sum\nolimits_{i=1}^{M} \sqrt{\alpha_{i} (t)}$, and the $l_2$ norm of the gradient vectors at different devices are not highly skewed.      
Accordingly, the model parameter vector in Algorithm \ref{ModelUpdateAlg} is updated as follows:
\begin{align}\label{ModelUpdateAfterAssump}
\boldsymbol{\theta}_{t+1} = \boldsymbol{\theta}_t - \eta \cdot {\rm{AMP}}_{\boldsymbol{A}_{s-1}} & \bigg(  \boldsymbol{A}_{s-1} \sum\nolimits_{m=1}^{M} \frac{ 1}{M} \boldsymbol{g}_m^{sp} \left(\boldsymbol{\theta}_t \right) \bigg.\nonumber \\
& \bigg.+ \frac{1}{\sum\nolimits_{i=1}^{M} \sqrt{\alpha_{i} (t)}} \boldsymbol{z}^{s-1} (t) \bigg).   
\end{align}

\begin{lemma}\label{LemAMPRecons}\cite{AMPBayatiMontIT,RamjiAMPIT_1,RamjiAMPIT_2} 
Consider reconstructing vector $\boldsymbol{x} \in \mathbb{R}^d$ with sparsity level $k$ and $\left\| \boldsymbol{x} \right\|^2 = P$ from a noisy observation $\boldsymbol{y} = \boldsymbol{A}_s \boldsymbol{x} + \boldsymbol{z}$, where $\boldsymbol{A}_s \in \mathbb{R}^{s \times d}$ is the measurement matrix, and $\boldsymbol{z} \in \mathbb{R}^s$ is the noise vector with each entry i.i.d. with $\mathcal{N} (0, \sigma^2)$. If $s > k$, the AMP algorithm reconstructs 
\begin{align}\label{ReconAMPLem}
\hat{\boldsymbol{x}} \triangleq {\rm{AMP}}_{\boldsymbol{A}_{s}} ({\boldsymbol{y}}) = \boldsymbol{x} + \sigma_{\tau} \boldsymbol{\omega},
\end{align}
where each entry of $\boldsymbol{\omega} \in \mathbb{R}^d$ is i.i.d. with $\mathcal{N}(0,1)$, and $\sigma^2_{\tau}$ decreases monotonically from $\sigma^2 + P$ to $\sigma^2$. That is, the noisy observation $\boldsymbol{y}$ is effectively transformed into $\hat{\boldsymbol{x}} = \boldsymbol{x} + \sigma \boldsymbol{\omega}$.    
\end{lemma}

\begin{assumption}\label{AssumpSparsePattern}
We assume that the sparsity pattern of vector $\sum\nolimits_{m=1}^{M} \boldsymbol{g}_m^{sp} \left(\boldsymbol{\theta}_t \right)$ is smaller than $s-1$, $\forall t$.  
\end{assumption}
If all the sparse vectors $\boldsymbol{g}_m^{sp} \left(\boldsymbol{\theta}_t \right)$, $\forall m \in [M]$, have the same sparsity pattern, Assumption \ref{AssumpSparsePattern} holds trivially since we set $k < s-1$. On the other hand, in general we can guarantee that Assumption \ref{AssumpSparsePattern} holds by setting $k \ll s$.

According to Lemma \ref{LemAMPRecons} and Assumption \ref{AssumpSparsePattern}, the model parameter update given in \eqref{ModelUpdateAfterAssump} can be rewritten as
\begin{align}\label{ModelUpdateAfterAssumpLem}
\boldsymbol{\theta}_{t+1} = \boldsymbol{\theta}_t - \eta & \left(\frac{1}{M} \sum\nolimits_{m=1}^{M} \boldsymbol{g}_m^{sp} \left(\boldsymbol{\theta}_t \right) + \sigma_{\omega} (t) \boldsymbol{w} (t) \right),   
\end{align}
where we define
\begin{align}\label{SigmaWBasicDef}
\sigma_{\omega} (t) \triangleq \frac{\sigma}{\sum\nolimits_{m=1}^{M} \sqrt{\alpha_{m} (t)}},  
\end{align}
and each entry of $\boldsymbol{\omega}(t) \in \mathbb{R}^d$ is i.i.d. with $\mathcal{N}(0,1)$.

\begin{corollary}\label{CorSparsek}
For vector $\boldsymbol{x} \in \mathbb{R}^d$, it is easy to verify that 
$\left\| \boldsymbol{x} - {\rm{sp}}_k (\boldsymbol{x}) \right\| \le \lambda \left\|  \boldsymbol{x} \right\|$,
where $\lambda \triangleq \sqrt{\frac{d-k}{d}}$, and the equality holds if all the entries of $\boldsymbol{x}$ have the same magnitude. 
\end{corollary}

\begin{lemma}\label{LemL2NormGaussianVec}
Consider a random vector $\boldsymbol{u} \in \mathbb{R}^d$ with each entry i.i.d. according to $\mathcal{N} (0, \sigma_u^2)$. For $\delta \in (0,1)$, we have ${\rm{Pr}} \left\{ \left\| \boldsymbol{u} \right\| \ge \sigma_u \rho (\delta) \right\} = \delta$,
where we define $\rho ({\delta}) \triangleq \left( 2 \gamma^{-1} \left( \Gamma(d/2) (1-\delta),d/2 \right) \right)^{1/2}$, and $\gamma \left( a, x \right)$ is the lower incomplete gamma function defined as:
\begin{align}\label{lower_incomplete_gamma_function}
\gamma \left( a,x \right) \triangleq \int_0^x {{\iota ^{a - 1}}{e^{ - \iota }}d\iota }, \quad x \ge 0, a > 0,   
\end{align}
and $\gamma^{-1} \left(a,y  \right)$ is its inverse, which returns the value of $x$, such that $\gamma \left( a,x \right) = y$, and $\Gamma \left( a \right)$ is the gamma function defined as
\begin{align}\label{gamma_function}
\Gamma \left( a\right) \triangleq \int_0^{\infty} {{\iota ^{a - 1}}{e^{ - \iota }}d\iota }, \quad a > 0.   
\end{align}
\end{lemma}
\begin{proof}
We highlight that $\left\| \boldsymbol{u} \right\|^2$ follows the Chi-square distribution, and we have \cite{BookTheoryOfStatistics}
\begin{align}\label{AppChiSquare}
{\rm{Pr}} \left\{ \left\| \boldsymbol{u} \right\|^2 \le u \right\} = \frac{\gamma(d/2, u/(2\sigma_u^2))}{\Gamma(d/2)}.     
\end{align}
Accordingly, it follows that 
\begin{align}\label{AppChiSquare2}
&{\rm{Pr}} \left\{ \left\| \boldsymbol{u} \right\| \ge \sigma_u \left( 2 \gamma^{-1} \left( \Gamma(d/2) (1-\delta),d/2 \right) \right)^{1/2} \right\} \nonumber\\
& \;\;\;={\rm{Pr}} \left\{ \left\| \boldsymbol{u} \right\|^2 \ge 2 \sigma_u^2  \gamma^{-1} \left( \Gamma(d/2) (1-\delta),d/2 \right)  \right\} = \delta.     
\end{align} \end{proof}

By choosing $\delta$ appropriately, we can guarantee, with probability arbitrary close to $1$, that $\left\| \boldsymbol{u} \right\| \le \sigma_u \rho ({\delta})$.

\begin{lemma}\label{LemVarianceSigmaOmega}
Having $\sigma_{\omega} (t)$ defined as in \eqref{SigmaWBasicDef}, by taking the expectation with respect to the gradients, we have, $\forall t$,
\begin{align}\label{EQLemVarianceSigmaOmega}
\mathbb{E} \left[ \sigma_{\omega} (t) \right] \le \frac{\sigma}{M \sqrt{P_t}} \left( \sigma_{{\rm{max}}} \bigg( \frac{1-\lambda^{t+1}}{1-\lambda} \bigg) G + 1 \right),    
\end{align}
where $\sigma_{{\rm{max}}} \triangleq \sqrt{\frac{d}{s-1}} +1$.
\end{lemma}
\begin{proof}
See Appendix \ref{AppLemVarianceSigmaOmega}.
\end{proof}

\begin{lemma}\label{LemExpectationDifferenceTheta}
Based on the results presented above, taking expectation with respect to the gradients yields
\begin{subequations}
\begin{align}\label{EQLemExpectationDifferenceTheta}
{\small{\mathbb{E} \left[ \left\| \eta \frac{1}{M} \sum\nolimits_{m=1}^{M} \left( \boldsymbol{g}_m \left( \boldsymbol{\theta}_t \right)  - \boldsymbol{g}_m^{sp} \left( \boldsymbol{\theta}_t \right) \right) - \sigma_{\omega}(t) \boldsymbol{w}(t) \right\| \right] \le \eta v(t),}}   
\end{align}
where we define, $\forall t \ge 0$,  
\begin{align}\label{DefExpectationDifferenceThetaMainSimplev_t}
v(t) \triangleq & \lambda \left(  \frac{(1+ \lambda)(1-\lambda^t)}{1-\lambda} +1 \right)G \nonumber\\
&+  \rho(\delta) \frac{\sigma}{M \sqrt{P_t}} \left( \sigma_{{\rm{max}}} \bigg( \frac{1-\lambda^{t+1}}{1-\lambda} \bigg) G + 1 \right),   
\end{align}
\end{subequations} 
\end{lemma}
\begin{proof}
See Appendix \ref{AppLemExpectationDifferenceTheta}. 
\end{proof}

\begin{figure*}[t!]
\centering
\begin{subfigure}{.5\textwidth}
  \centering
  \includegraphics[scale=0.635,trim={20pt 7pt 36pt 40pt},clip]{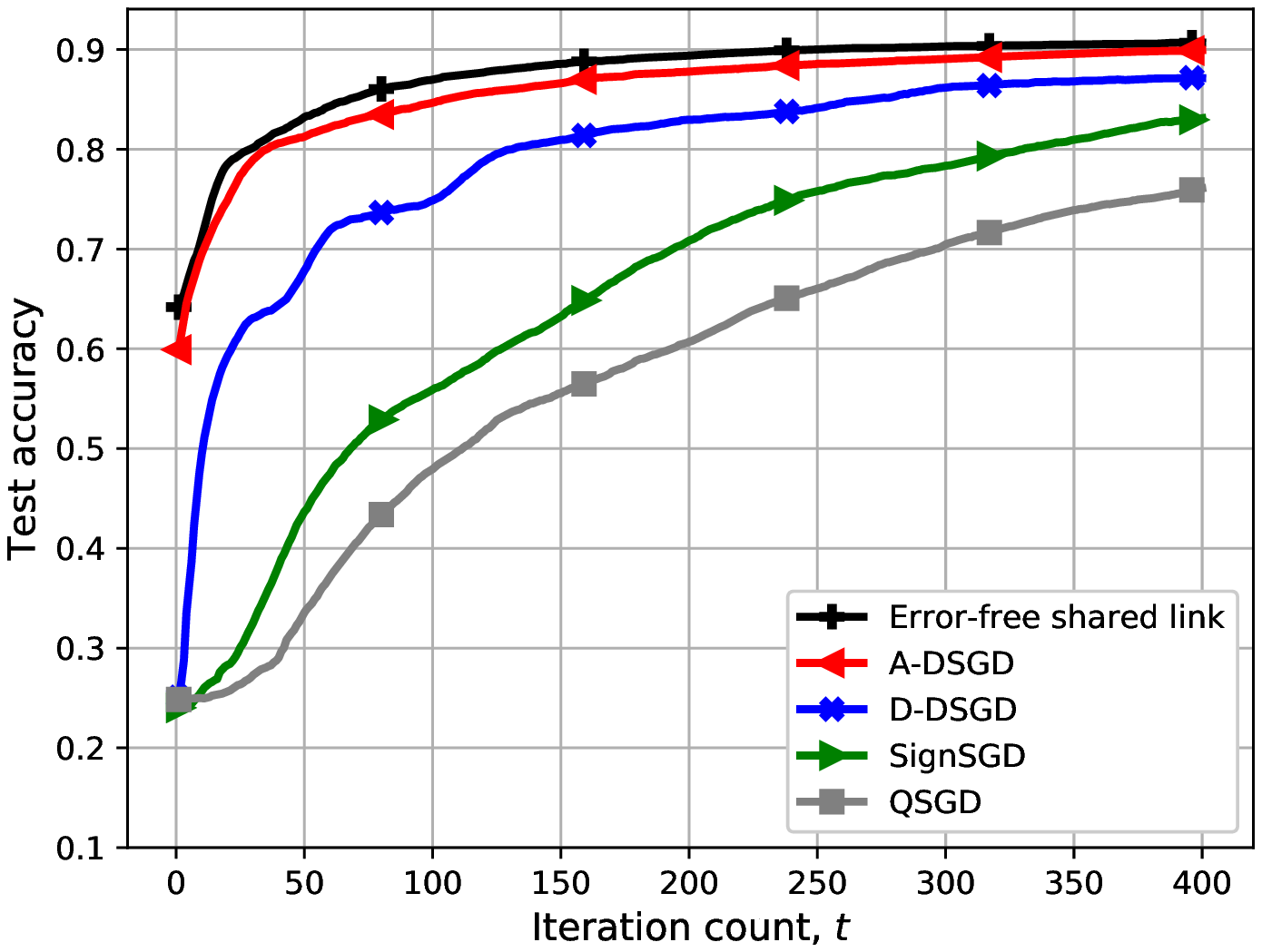}
  \caption{IID data distribution}
  \label{Fig_ADSGD_DDSGS_IID}
\end{subfigure}%
\begin{subfigure}{.5\textwidth}
  \centering
  \includegraphics[scale=0.635,trim={20pt 7pt 36pt 40pt},clip]{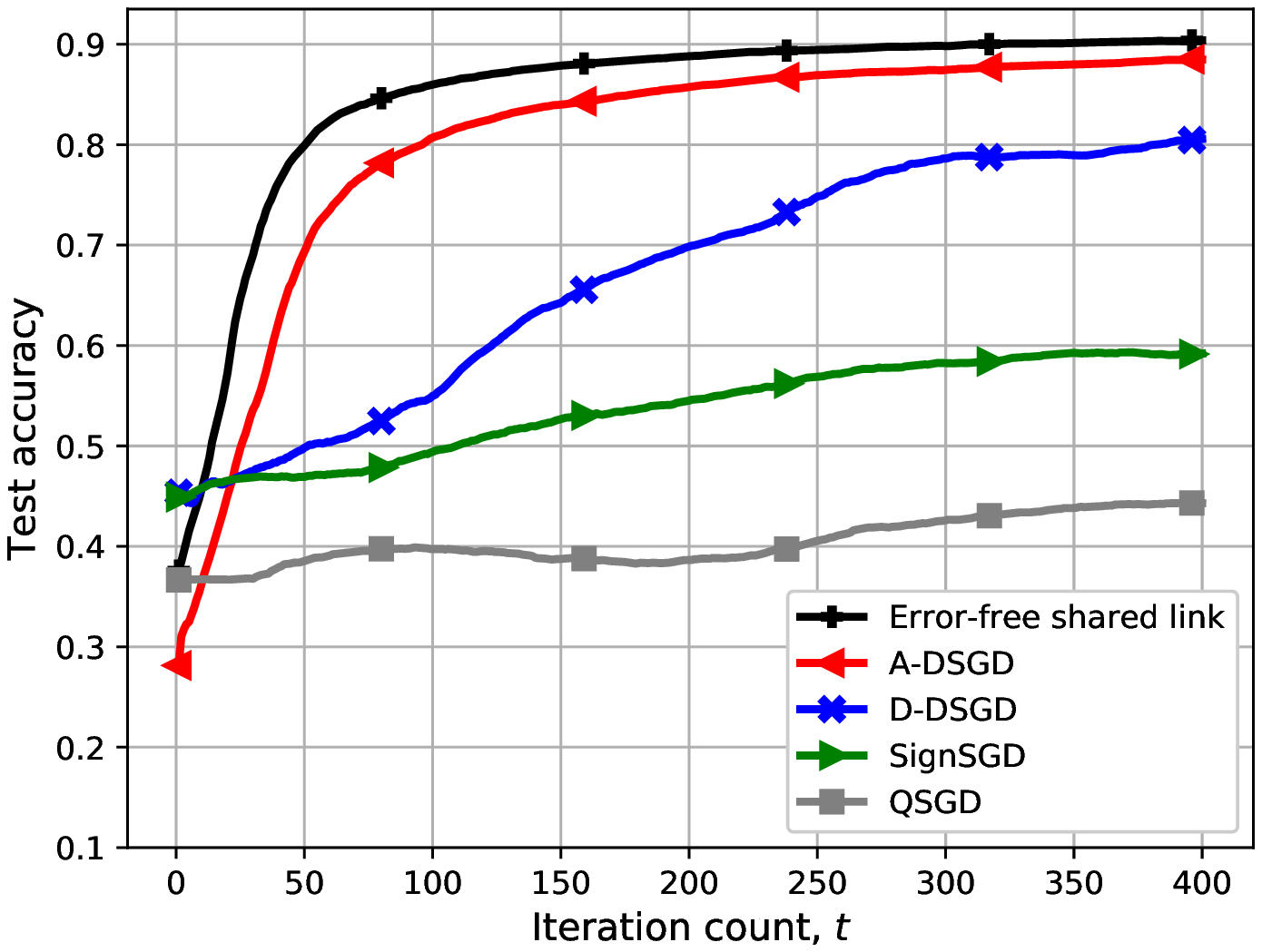}
  \caption{Non-IID data distribution}
  \label{Fig_ADSGD_DDSGS_NonIID}
\end{subfigure}
\caption{Test accuracy of different schemes for IID and non-IID data distribution scenarios with $M=25$, $B=1000$, $\bar{P} = 500$, $s=d/2$, $k=s/2$, and $P_t=\bar{P}$, $\forall t$.}
\label{Fig_ADSGD_DDSGS_IID_NonIID}
\end{figure*}

\subsection{Convergence Rate}\label{ConvergRate}
Let $\boldsymbol{\theta}^*$ be the optimum parameter vector minimizing the loss function $F$. For $\varepsilon > 0$, we define $\mathcal{S} \triangleq \{ \boldsymbol{\theta}: \left\| \boldsymbol{\theta} - \boldsymbol{\theta}^* \right\|^2 \le \varepsilon \}$ as the \textit{success region} to which the parameter vector should converge. We establish the convergence using rate supermartingale stochastic processes (refer to \cite{SaTamingWildNips15,SparseConvergenceAlistarh} for a detailed definition of the supermartingale and rate supermartingale processes).

\begin{definition}\cite{SaTamingWildNips15}
For any stochastic process, a non-negative process $W_t: \mathbb{R}^{d \times (t+1)} \to \mathbb{R}$ with horizon $B$ is a rate supermartingale if i) for any sequence $\boldsymbol{\theta}_t, \dots, \boldsymbol{\theta}_0 \in \mathbb{R}^d$ and any $t \le B$,
$\mathbb{E} \left[ W_{t+1} \left( \boldsymbol{\theta}_{t+1}, \dots, \boldsymbol{\theta}_0 \right) \right] \le W_{t} \left( \boldsymbol{\theta}_{t}, \dots, \boldsymbol{\theta}_0 \right)$.
ii) $\forall T \le B$ and any sequence $\boldsymbol{\theta}_T, \dots, \boldsymbol{\theta}_0 \in \mathbb{R}^{d}$, if the algorithm does not enter the success region by time $T$, i.e., if $\boldsymbol{\theta}_t \notin \mathcal{S}$, $\forall t \le T$, it must satisfy $W_{T} \left( \boldsymbol{\theta}_{T}, \dots, \boldsymbol{\theta}_0 \right) \ge T$.
\end{definition}

Next we present a rate supermartingale process for the stochastic process with the update in \eqref{ParallelSGDModelUpdate}, introduced in \cite{SaTamingWildNips15}.  

\vspace{.1cm}
\noindent \textbf{Statement 1.} Consider the stochastic process in \eqref{ParallelSGDModelUpdate} with learning rate $\eta < 2 c \varepsilon / G^2$. We define the process $W_t$ by
\begin{align}\label{ProcessW_tHasNotSucceesed}
{\small{W_t (\boldsymbol{\theta}_t, ..., \boldsymbol{\theta}_0) \triangleq \frac{\varepsilon}{2\eta c \varepsilon - \eta^2 G^2} \log \bigg( \frac{e \left\| \boldsymbol{\theta}_t - \boldsymbol{\theta}^* \right\|^2}{\varepsilon} \bigg) +t,}} 
\end{align}
if the algorithm has not entered the success region by iteration $t$, i.e., if $\boldsymbol{\theta}_i \notin \cal{S}$, $\forall i \le t$, and $W_t (\boldsymbol{\theta}_t, ..., \boldsymbol{\theta}_0) = W_{\tau_1-1} (\boldsymbol{\theta}_{\tau_1-1}, ..., \boldsymbol{\theta}_0)$ otherwise, where ${\tau_1}$ is the smallest index for which $\boldsymbol{\theta}_{\tau} \in \cal{S}$. The process $W_t$ is a rate supermartingale with horizon $B = \infty$. It is also $L$-Lipschitz in its first coordinate, where $L \triangleq 2 \sqrt{\varepsilon} (2\eta c \varepsilon - \eta^2 G^2)^{-1}$; that is, given any $\boldsymbol{x}, \boldsymbol{y} \in \mathbb{R}^d, t \ge 1$, and any sequence $\boldsymbol{\theta}_{t-1}, ..., \boldsymbol{\theta}_0$, we have  
\begin{align}\label{EQFirstCoordonate}
{\footnotesize{\left\| W_t \left(\boldsymbol{y}, \boldsymbol{\theta}_{t-1}, ..., \boldsymbol{\theta}_0 \right) - W_t \left(\boldsymbol{x}, \boldsymbol{\theta}_{t-1}, ..., \boldsymbol{\theta}_0 \right) \right\| \le L \left\| \boldsymbol{y} - \boldsymbol{x} \right\|.}}     
\end{align}

\begin{theorem}\label{TheoremMainProbError}
Under the assumptions described above, for
\begin{align}\label{EtaBoundTheorem}
\eta < \frac{2\left( c \varepsilon T - \sqrt{\varepsilon} \sum\nolimits_{t=0}^{T-1} v(t) \right)}{ T G^2},    
\end{align}
the probability that the A-DSGD algorithm does not enter the success region by time $T$ is bounded by
\begin{align}\label{EQTheoremMainProbError}
{\small{{\rm{Pr}} \left\{ E_T \right\} \le \frac{\varepsilon}{(2\eta c \varepsilon - \eta^2 G^2) \left(T-\eta L \sum\nolimits_{t=0}^{T-1}v(t) \right)} \log \bigg( \frac{e \left\| \boldsymbol{\theta}^* \right\|^2}{\varepsilon} \bigg),}} 
\end{align}
where $E_T$ denotes the event of not arriving at the success region at time $T$.
\end{theorem}
\begin{proof}
See Appendix \ref{AppTheoremMainProbError}.  
\end{proof} 
 
We highlight that the reduction term $\sum\nolimits_{t=0}^{T-1} v(t)$ in the denominator is due to the sparsification, random projection (compression), and the noise introduced by the wireless channel.  
To show that ${\rm{Pr}} \left\{ E_T \right\} \to 0$ asymptotically, we consider a special case $P_t = \bar{P}$, in which 
{\small \begin{align}\label{SumVtPtBarP}
& \sum\nolimits_{t=0}^{T-1} v(t) =  \left( \frac{2\lambda G}{1 - \lambda} + \frac{\sigma \rho(\delta)}{M \sqrt{\bar{P}}} \left( \frac{ \sigma_{\rm{max}}G}{1-\lambda} +1 \right)  \right) T \nonumber\\
& - \left( \frac{\lambda (1+\lambda)(1-\lambda^T) G}{(1-\lambda)^2} + \frac{\sigma \rho(\delta) \sigma_{\rm{max}}(1-\lambda^{T+1})G}{M \sqrt{\bar{P}} (1-\lambda)^2} \right).  
\end{align}}
After replacing $\sum\nolimits_{t=0}^{T-1} v(t)$ in \eqref{EQTheoremMainProbError}, it is easy to verify that, for $\eta$ bounded as in \eqref{EtaBoundTheorem}, ${\rm{Pr}} \left\{ E_T \right\} \to 0$ as $T \to \infty$.

\section{Experiments}\label{SecExperiments}
Here we evaluate the performances of A-DSGD and D-DSGD for the task of image classification. We run experiments on the MNIST dataset \cite{LeCunMNIST} with $N = 60000$ training and $10000$ test data samples, and train a single layer neural network with $d=7850$ parameters utilizing ADAM optimizer \cite{ADAMDC}. We set the channel noise variance to $\sigma^2=1$. The performance is measured as the accuracy with respect to the test dataset, called \textit{test accuracy}, versus iteration count $t$.

We consider two scenarios for data distribution across devices: in \textit{IID data distribution}, $B$ training data samples, selected at random, are assigned to each device at the beginning of training; while in \textit{non-IID data distribution}, each device has access to $B$ training data samples, where $B/2$ of the training data samples are selected at random from only one class/label; that is, for each device, we first select two classes/labels at random, and then randomly select $B/2$ data samples from each of the two classes/labels. At each iteration, devices use all the $B$ local data samples to compute their gradient estimates, i.e., the batch size is equal to the size of the local datasets.

When performing A-DSGD, we use mean removal power allocation scheme only for the first 20 iterations. The reason is that the gradients are more aggressive at the beginning of A-DSGD, and their mean values are expected to be relatively diverse, while they converge to zero.  
Also, for comparison, we consider a benchmark, where the PS receives the average of the actual gradient estimates at the devices, $\frac{1}{M} \sum\nolimits_{m=1}^{M} \boldsymbol{g}_m \left(\boldsymbol{\theta}_t \right)$, in a noiseless fashion, and updates the model parameter vector according to this error-free observation. We refer to this scheme as the error-free shared link approach. 

\begin{figure}[t!]
\centering
\includegraphics[scale=0.635,trim={20pt 7pt 45pt 40pt},clip]{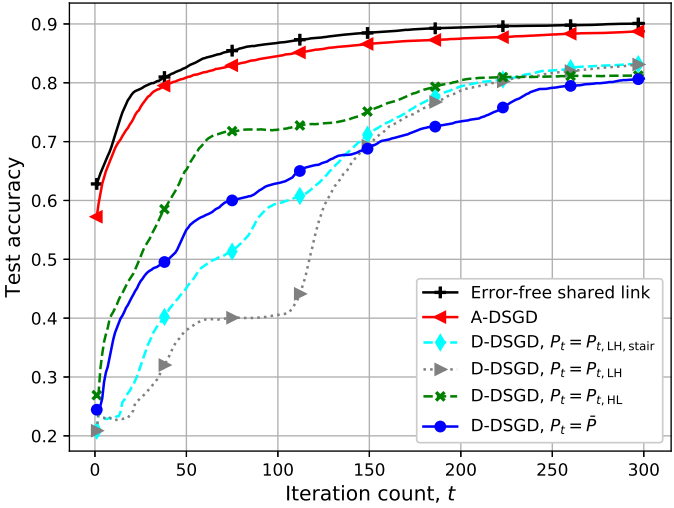}
\caption{Performance of the A-DSGD and D-DSGD algorithms for different power allocation schemes when $M=25$, $B=1000$, $\bar{P}=200$, $s=d/2$, and $k=s/2$.}
\label{ADSGD_DDSGD_power_allocation}
\end{figure}

In Fig. \ref{Fig_ADSGD_DDSGS_IID_NonIID}, we investigate IID and non-IID data distribution scenarios. For comparison, we also consider two alternative digital schemes that employ SignSGD \cite{SignSGDBernstein} and QSGD \cite{QSGDQuantDAlistarh} algorithms for gradient compression, where each device applies the coding schemes proposed in \cite{SignSGDBernstein} and \cite{QSGDQuantDAlistarh} to a limited number of the elements of its gradient vector, specified by its link capacity. To be more precise, assume that, at iteration $t$, at each device a limited number of $q_{t, \rm{S}}$ and $q_{t, \rm{Q}}$ entries of the gradient estimate with highest magnitudes are selected for delivery with SignSGD and QSGD, respectively. Following the SignSGD scheme \cite{SignSGDBernstein}, at iteration $t$, a total number of 
\begin{align}\label{SignSGDDeliveredBits}
r_{t, \rm{S}} = \log_2 \binom{d}{q_{t, \rm{S}}} + {q_{t, \rm{S}}} \; \mbox{bits}, \quad \forall t,    
\end{align}
are delivered to transmit the sign of each selected entry, as well as their locations, and ${q_{t, \rm{S}}}$ is then highest integer satisfying $r_{t, \rm{S}} \le R_t$. On the other hand, following the QSGD approach \cite{QSGDQuantDAlistarh}, each device sends a quantized version of each of the $q_{t, \rm{Q}}$ selected entries, the $l_2$-norm of the resultant vector with $q_{t, \rm{Q}}$ non-zero entries, and the locations of the non-zero entries. Hence, for a quantization level of $2^{l_{\rm{Q}}}$, each device transmits 
\begin{align}\label{QSGDDeliveredBits}
r_{t, \rm{Q}} = 32 + \log_2 \binom{d}{q_{t, \rm{Q}}} + (1+l_{\rm{Q}}) {q_{t, \rm{Q}}} \; \mbox{bits}, \quad \forall t,    
\end{align}
where ${q_{t, \rm{Q}}}$ is set as then highest integer satisfying $r_{t, \rm{Q}} \le R_t$. For QSGD, we consider a quantization level $l_{\rm{Q}}=2$.  
We consider $s=d/2$ channel uses, $M=25$ devices each with $B=1000$ training data samples, average power $\bar{P}=500$, and we set a fixed ratio $k= \left\lfloor {s/2} \right\rfloor$ for sparsification, and $P_t = \bar{P}$, $\forall t$. Observe that A-DSGD outperforms all digital approaches that separate computation from communication. For both data distribution scenarios, the gap between the performance of the error-free shared link approach and that of A-DSGD is very small, and D-DSGD significantly outperforms the other two digital schemes SignSGD and QSGD. Unlike the digital schemes, the performance loss of A-DSGD in non-IID data distribution scenario compared to the IID scenario is negligible illustrating the robustness of A-DSGD to the bias in data distribution. The reduction in the convergence rate of A-DSGD in the non-IID scenario is mainly due to the mismatch in the sparsity patterns of gradients at different devices, especially in the initial iterations. We also observe that the performance of D-DSGD degrades much less compared to SignSGD and QSGD when the data is biased across devices.

\begin{figure}[t!]
\centering
\includegraphics[scale=0.635,trim={20pt 7pt 45pt 40pt},clip]{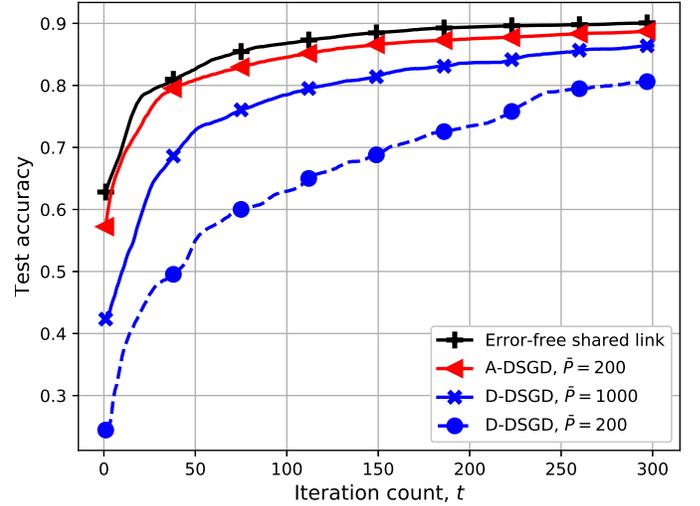}
\caption{Performance of the A-DSGD and D-DSGD algorithms for different $\bar{P}$ values $\bar{P} \in \{200, 1000\}$, when $M=25$, $B=1000$, $s=d/2$, $k=s/2$, and $P_t = \bar{P}$, $\forall t$.}
\label{UPA_SA_SD_DifPower}
\end{figure}

In all the following experiments, we only consider IID data distribution scenario. In Fig. \ref{ADSGD_DDSGD_power_allocation} we investigate the performance of D-DSGD for various power allocation designs for $T=300$ DSGD iterations with $M=25$, $B=1000$, $\bar{P}=200$, $s=d/2$, and $k=\left\lfloor {s/2} \right\rfloor$. We consider four different power allocation schemes: constant power allocation $P_t = \bar{P}$, and
\begin{subequations}\label{PowerAllocationApproaches}
\begin{align}\label{PowerAllocationApproaches_stair}
{P_{t,{\rm{LH, stair}}}} &= 100 \Big( \frac{2}{299}(t-1)+1 \Big), \quad t \in [300],\\
{P_{t,{\rm{LH}}}} &= \begin{cases} 
100, & \mbox{if $1 \le t \le 100$},\\
200, & \mbox{if $101 \le t \le 200$},\\ 
300, & \mbox{if $201 \le t \le 300$},
\end{cases}\label{PowerAllocationApproaches_LH}\\
{P_{t,{\rm{HL}}}} &= \begin{cases} 
300, & \mbox{if $1 \le t \le 100$},\\
200, & \mbox{if $101 \le t \le 200$},\\ 
100, & \mbox{if $201 \le t \le 300$}.\label{PowerAllocationApproaches_HL}
\end{cases}
\end{align}
\end{subequations}
With ${P_{t,{\rm{LH, stair}}}}$, average transmit power increases linearly from $P_1=100$ to $P_{300}=300$. We also consider the error-free shared link bound, and A-DSGD with $P_t = \bar{P}$, $\forall t$. Note that the performance of A-DSGD is almost the same as in Fig. \ref{Fig_ADSGD_DDSGS_IID}, despite $60\%$ reduction in average transmit power. The performance of D-DSGD falls short of A-DSGD for all the power allocation schemes under consideration. Comparing the digital scheme for various power allocation designs, we observe that it is better to allocate less power to the initial iterations, when the gradients are more aggressive, and save the power for the later iterations to increase the final accuracy; even though this will result in a lower convergence speed. Also, increasing the power slowly, as in ${P_{t,{\rm{LH, stair}}}}$, provides a better performance in general, compared to occassional jumps in the transmit power. Overall, letting $P_t$ vary over time provides an additional degree-of-freedom, which can improve the accuracy of D-DSGD by designing an efficient power allocation scheme having statistical knowledge about the dynamic of the gradients over time.            

In Fig. \ref{UPA_SA_SD_DifPower}, we compare the performance of A-DSGD with that of D-DSGD for different values of the available average transmit power $\bar{P}\in \{ 200, 1000\}$. We consider $s=d/2$, $M=25$, and $B=1000$. We set $k= \left\lfloor {s/2} \right\rfloor$ and $P_t = \bar{P}$, $\forall t$. As in Fig. \ref{Fig_ADSGD_DDSGS_IID_NonIID}, we observe that A-DSGD outperforms D-DSGD, and the gap between A-DSGD and the error-free shared link approach is relatively small. We did not include the performance of A-DSGD for $\bar{P} = 1000$ since it is very close to the one with $\bar{P} = 200$. Unlike A-DSGD, the performance of D-DSGD significantly deteriorates by reducing $\bar{P}$. Thus, we conclude that the analog approach is particularly attractive for learning across low-power devices as it allows them to align their limited transmission powers to dominate the noise term.

In Fig. \ref{UPA_SA_SD_Difdevices}, we compare A-DSGD and D-DSGD for different channel bandwidth values, $s \in \{ d/2, 3d/10 \}$, and $M=20$ and $B=1000$, where we set $P_t = 500$, $\forall t$, and $k = \left\lfloor {s/2} \right\rfloor$. Performance deterioration of D-DSGD is notable when channel bandwidth reduces from $s = d/2$ to $s=3 d/10$. On the other hand, we observe that A-DSGD is robust to the channel bandwidth reduction as well, providing yet another advantage of analog over-the-air computation for edge learning. 

\begin{figure}[t!]
\centering
\includegraphics[scale=0.635,trim={20pt 7pt 45pt 40pt},clip]{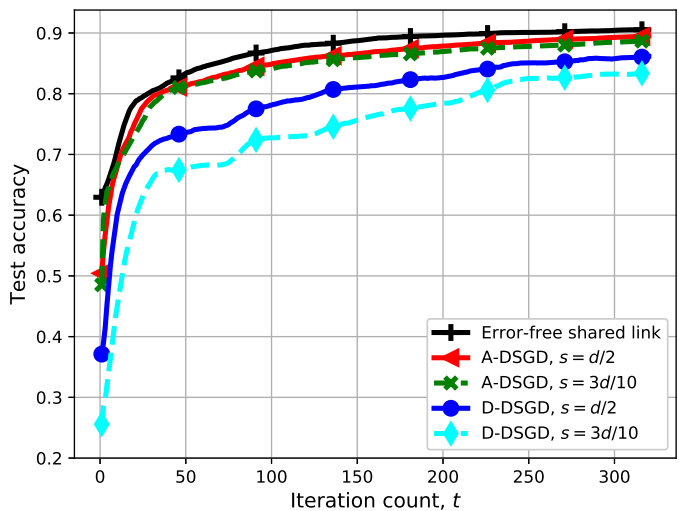}
\caption{Performance of the A-DSGD and D-DSGD algorithms for different $s$ values, $s \in \{ d/2, 3d/10 \}$, when $M=20$, $B=1000$, $\bar{P}=500$, and $k = \left\lfloor {s/2} \right\rfloor$, where $P_t = \bar{P}$, $\forall t$.}
\label{UPA_SA_SD_Difdevices}
\end{figure}

\begin{figure}[t!]
\centering
\includegraphics[scale=0.635,trim={20pt 7pt 45pt 40pt},clip]{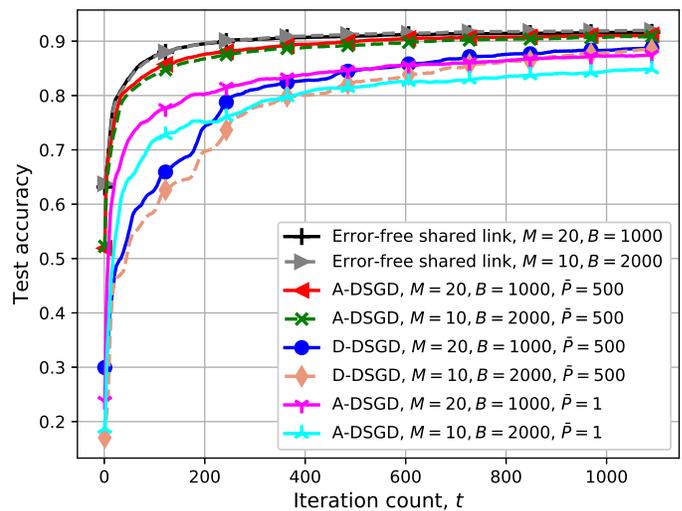}
\caption{Performance of the A-DSGD and D-DSGD algorithms for different $\left( M,B \right)$ pairs $(M,B) \in \{(10,2000), (20,1000)\}$, where $s =d/2$, $k = s/2$, and $P_t = \bar{P}$, $\forall t$.}
\label{ADSGD_DDSGD_M_B}
\end{figure}

\begin{figure*}[t!]
\centering
\begin{subfigure}{.5\textwidth}
  \centering
  \includegraphics[scale=0.62,trim={13pt 7pt 45pt 40pt},clip]{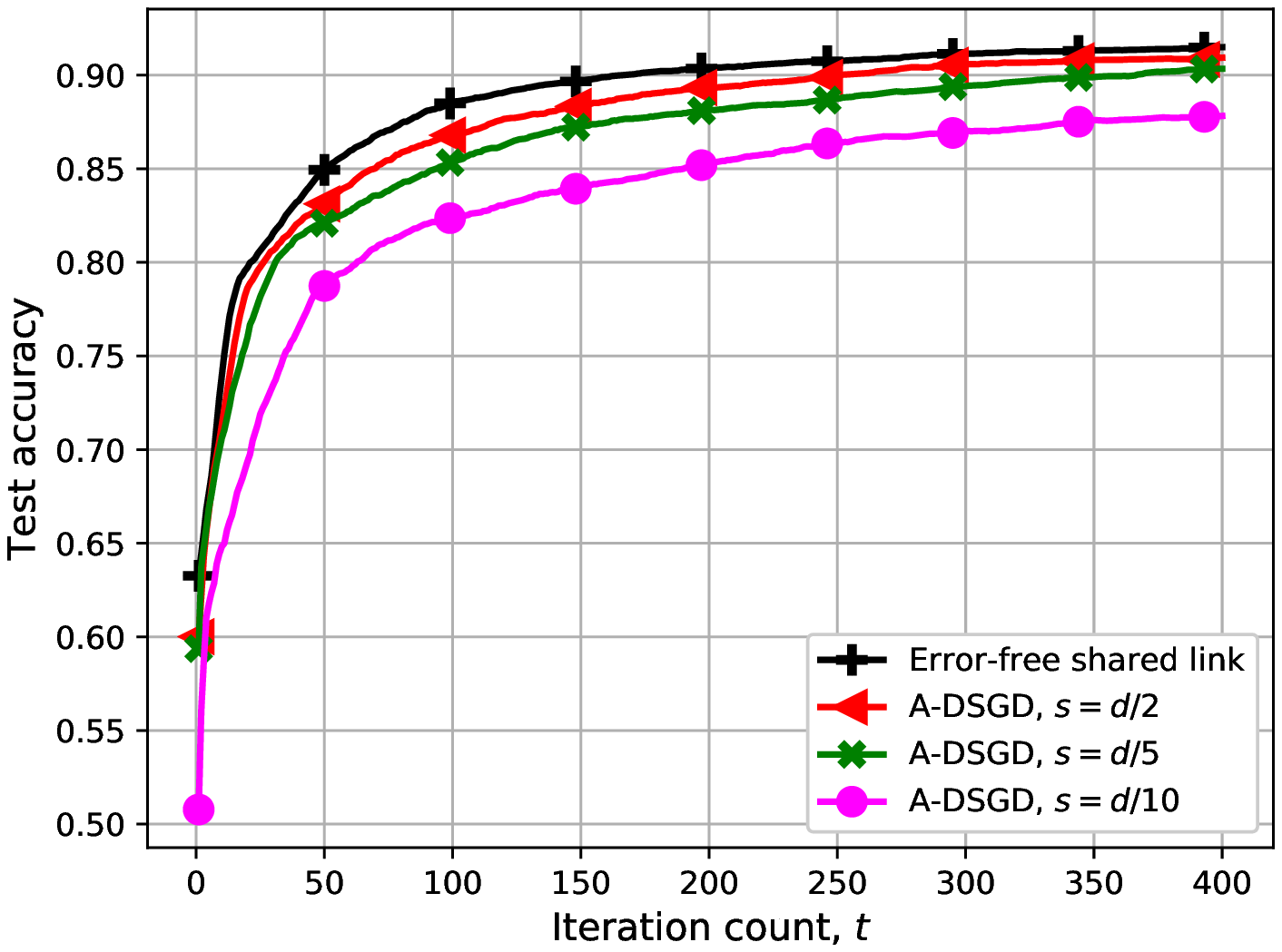}
  \caption{Test accuracy versus iteration count, $t$}
  \label{EPA_UPA_SA_Err_VsIter}
\end{subfigure}%
\begin{subfigure}{.5\textwidth}
  \centering
  \includegraphics[scale=0.62,trim={13pt 7pt 43pt 40pt},clip]{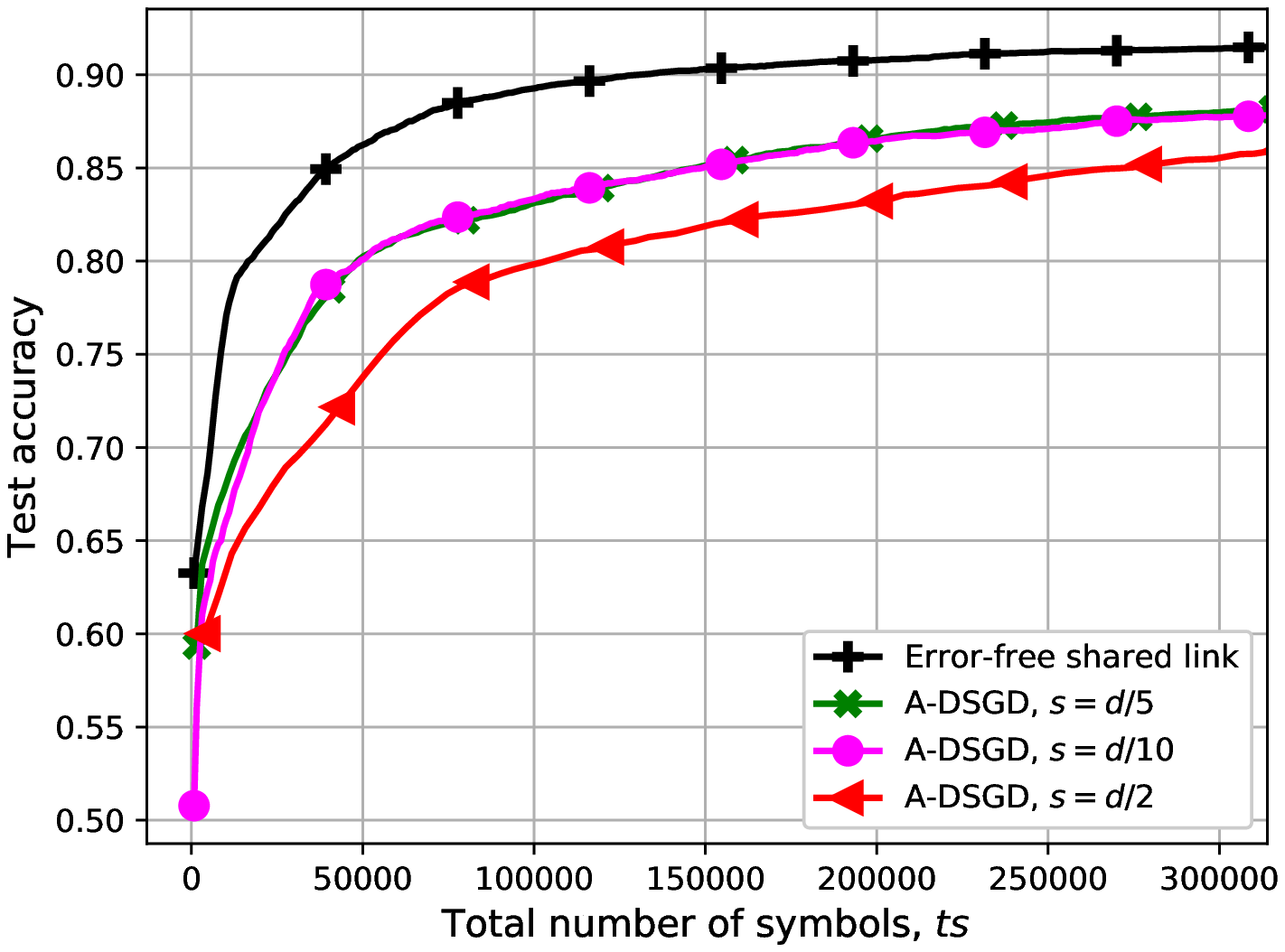}
  \caption{Test accuracy versus number of transmitted symbols, $t s$}
  \label{EPA_UPA_SA_Acc_VsIter}
\end{subfigure}
\caption{Performance of the A-DSGD algorithm for different $s$ values $s \in \{ d/10, d/5, d/2 \}$, when $M=25$, $B=1000$, $\bar{P}=50$, and $k = \left\lfloor {4s/5} \right\rfloor$, where $P_t = \bar{P}$, $\forall t$.}
\label{EPA_UPA_SA_VsIter}
\end{figure*}

In Fig. \ref{ADSGD_DDSGD_M_B}, we compare the performance of A-DSGD and D-DSGD for different $M$ and $\bar{P}$ values, while keeping the total number of data samples available across all the devices, $MB$, fixed. We consider $(M,B) \in \left\{(10, 2000), (20, 1000)\right\}$, and, for each setting, $\bar{P} \in \{ 1,500 \}$. Given $s=\left\lfloor {d/4} \right\rfloor$, we set $k=\left\lfloor {s/2} \right\rfloor$, and $P_t = \bar{P}$, $\forall t$. We highlight that, for $\bar{P} = 1$, D-DSGD fails since the devices cannot transmit any information bits. We observe that increasing $M$ while fixing $MB$ does not have any visible impact on the performance in the error-free shared link setting. The convergence speed of A-DSGD improves with $M$ for both $\bar{P}$ values. We note that although the accuracy of individual gradient estimates degrade with $M$ (due to the reduction in the training sample size available at each device), A-DSGD benefits from the additional power introduced by each device, which increases the robustness against noise thanks to the superposed signals each with average power $\bar{P}$. Similarly, the convergence speed of D-DSGD improves with $M$ for $\bar{P}=500$. Note that even though the number of bits allocated to each device reduces with $M$, i.e., each device sends a less accurate estimate of its gradient, since each new device comes with its own additional power, the total number of bits transmitted per data sample increases. 
Also, we observe that, for $\bar{P} = 500$, the performance of A-DSGD improves slightly by increasing $M$, while for $\bar{P} = 1$, the improvement is notable. Thus, even for large enough $\bar{P}$, increasing $M$ slightly improves the performance of A-DSGD. 

In Fig. \ref{EPA_UPA_SA_VsIter}, we investigate the impact of $s$ on the performance of A-DSGD. We consider $s \in \{ d/10, d/5, d/2 \}$ for a distributed system with $M=25$, $B = 1000$, $\bar{P} = 50$, and for any $s$ values, we set $k = \left\lfloor {4s/5} \right\rfloor$, and $P_t = \bar{P}$, $\forall t$. In Fig. \ref{EPA_UPA_SA_Err_VsIter}, we consider different $s$ values for each iteration of A-DSGD. This may correspond to allocating different number of channel resources for each iteration, for example, allocating more sub-bands through OFDM. On the other hand, in Fig. \ref{EPA_UPA_SA_Acc_VsIter}, we limit the number of channel uses at each communication round, and assume that the transmission time at each iteration linearly increases with $s$, and evaluate the performance of A-DSGD with respect to $ts$, the total number of transmitted symbols. This would correspond to the convergence time if these symbols are transmitted over time, using a single subcarrier. Therefore, assigning more resources to each iteration, i.e., increasing $s$, means that less number of iterations can be implemented by the same time. Observe from Fig. \ref{EPA_UPA_SA_Err_VsIter} that, as expected, utilizing more channel uses improves the performance both in terms of convergence speed and final accuracy. 
On the other hand, when we consider the total number of transmitted symbols in Fig. \ref{EPA_UPA_SA_Acc_VsIter} the performance of A-DSGD improves significantly by reducing $s$ from $s=d/2$ to $s = d/5$, which shows that performing more SGD iterations with less accurate gradient estimates at each iteration, i.e., smaller $s$ values, can be better. This clearly shows the benefit of A-DSGD, in which the devices can transmit low dimensional gradient vectors at each iteration instead of sending the entire gradient vectors over several iterations.
However, this trend does not continue, and the performance of A-DSGD almost remains unchanged by further reducing $s$ from $s=d/5$ to $s=d/10$. 
According to the results illustrated in Fig. \ref{EPA_UPA_SA_Acc_VsIter}, we conclude that, when the communication bandwidth is limited, the value $s$ can be a design parameter taking into account the total desired time (number of iterations) to finish an ML task through DSGD. We note here that we do not consider the computation time, which remain mostly the same independent of $s$ in different devices, as the same number of gradients are computed. 

\section{Conclusions}\label{SecConc}
We have studied collaborative/ federated learning at the wireless edge, where wireless devices aim to minimize an empirical loss function collaboratively with the help of a remote PS. Devices have their own local datasets, and they communicate with the PS over a wireless MAC. We have assumed that the communication from the PS to the devices is noiseless, so the updated parameter vector is shared with the devices in a lossless fashion. 

As opposed to the standard approach in FL, which ignores the channel aspects, and simply aims at reducing the communication load by compressing the gradients to a fixed level, here we incorporate the wireless channel characteristics and constraints into the system design. We considered both a digital approach (D-DSGD) that separates computation and communication, and an analog approach (A-DSGD) that exploits the superposition property of the wireless channel to have the average gradient computed over-the-air.

In D-DSGD, the amount of information bits sent by each device at each iteration can be adaptively adjusted with respect to the average transmit power constraint $\bar{P}$, and each device quantizes its gradient followed by capacity-achieving channel coding. We have shown that, under a finite power constraint, the performance of D-DSGD can be improved by allocating more power to later iterations to compensate for the decreasing gradient variance, and provide better protection against noise.

In A-DSGD, we have proposed gradient sparsification followed by random linear projection employing the same projection matrix at all the devices. This allowed reducing the typically very large parameter vector dimension to the limited channel bandwidth. The devices then transmit these compressed gradient vectors simultaneously over the MAC to exploit the superposition property of the wireless medium. This analog approach allows a much more efficient use of the available limited channel bandwidth. 

Numerical results have shown significant improvement in the performance with the analog approach, particularly in the low-power and low-bandwidth regimes. Moreover, when non-IID data distribution is considered, the performance loss of both A-DSGD and D-DSGD is much smaller compared to SignSGD \cite{SignSGDBernstein} and QSGD \cite{QSGDQuantDAlistarh}, A-DSGD being more robust to bias in data distribution across devices. We have also observed that the performances of both D-DSGD and A-DSGD improve as more devices are introduced thanks to the additional power introduced by each device. 



\appendices

\section{Proof of Lemma \ref{LemVarianceSigmaOmega}}\label{AppLemVarianceSigmaOmega}
By plugging \eqref{ScaleFacCalcUPA2} into \eqref{SigmaWBasicDef}, we have
\begin{align}\label{EQAppLemVarianceSigmaOmega_1}
\sigma_{\omega} (t) &= \sigma \bigg( \sum\limits_{m=1}^{M} \frac{\sqrt{P_t}}{\sqrt{\left\| \tilde{\boldsymbol{g}}_m \left( \boldsymbol{\theta}_{t} \right) \right\|^2 + 1}} \bigg)^{-1} \nonumber\\
& \le \sigma \bigg( \sum\limits_{m=1}^{M} \frac{\sqrt{P_t}}{\left\| \tilde{\boldsymbol{g}}_m \left( \boldsymbol{\theta}_{t} \right) \right\| + 1} \bigg)^{-1}.     
\end{align}
Here we find an upper bound on $\left\| \tilde{\boldsymbol{g}}_m \left( \boldsymbol{\theta}_{t} \right) \right\|$. We have, $\forall m,t$, 
\begin{align}\label{EQAppLemVarianceSigmaOmega_2}
\left\| \tilde{\boldsymbol{g}}_m \left( \boldsymbol{\theta}_{t} \right) \right\| & = \left\| \boldsymbol{A}_{s-1} {\boldsymbol{g}}^{\rm{sp}}_m \left( \boldsymbol{\theta}_{t} \right) \right\| \nonumber\\
& \le \left\| \boldsymbol{A}_{s-1} \right\| \left( \left\| {\boldsymbol{g}}_m \left( \boldsymbol{\theta}_{t} \right) \right\| + \left\| {\boldsymbol{\Delta}}_m \left( t \right) \right\| \right).     
\end{align}
We note that $\left\| \boldsymbol{A}_{s-1} \right\|$ corresponds to the largest singular value of random matrix $\boldsymbol{A}_{s-1}$. We use the result presented in \cite[Theorem 2]{LimitSmallEigenBaiYin} to estimate $\left\| \boldsymbol{A}_{s-1} \right\|$. According to \cite[Theorem 2]{LimitSmallEigenBaiYin}, for a random matrix $\boldsymbol{A} \in \mathbb{R}^{n \times m}$ with the $(i,j)$-th entry i.i.d. with $\mathbb{E} \left[ a_{i,j} \right] = 0$ and $\mathbb{E} \left[ a_{i,j}^2 \right] = 1/n^2$, $i \in [n]$, $j \in [m]$, and $n = \kappa m$, as $m \to \infty$, the largest singular value of $\boldsymbol{A}$ is asymptotically given by $\sigma_{\rm{max}} (\boldsymbol{A}) = \sqrt{m/n}+1$. Based on this result, since $d \gg 0$, we estimate $\left\| \boldsymbol{A}_{s-1} \right\|$ with its asymptotic largest singular value denoted by $\sigma_{\rm{max}} = \sqrt{d/(s-1)}+1$. Next we upper bound $\left\| {\boldsymbol{\Delta}}_m \left( t \right) \right\|$. According to Corollary \ref{CorSparsek}, we have, $\forall m,t$, 
\begin{align}\label{EQAppLemVarianceSigmaOmega_3}
\left\| {\boldsymbol{\Delta}}_m \left( t \right) \right\| &\le \lambda \left\| {\boldsymbol{g}}_m \left( \boldsymbol{\theta}_{t-1} \right) + {\boldsymbol{\Delta}}_m \left( t-1 \right) \right\| \nonumber\\
& \le \lambda (\left\| {\boldsymbol{g}}_m \left( \boldsymbol{\theta}_{t-1} \right) \right\| +  \left\|{\boldsymbol{\Delta}}_m \left( t-1 \right) \right\|), 
\end{align}
where iterating it downward yields
\begin{align}\label{EQAppLemVarianceSigmaOmega_4}
\left\| {\boldsymbol{\Delta}}_m \left( t \right) \right\| \le \sum\nolimits_{i=0}^{t-1} \lambda^{t-i} \left\| {\boldsymbol{g}}_m \left( \boldsymbol{\theta}_{i} \right) \right\|.  
\end{align}
By replacing $\left\| \boldsymbol{A}_{s-1} \right\|$ with $\sigma_{\rm{max}}$ and plugging \eqref{EQAppLemVarianceSigmaOmega_4} into \eqref{EQAppLemVarianceSigmaOmega_2}, it follows
\begin{align}\label{EQAppLemVarianceSigmaOmega_5}
\left\| \tilde{\boldsymbol{g}}_m \left( \boldsymbol{\theta}_{t} \right) \right\| \le \sigma_{\rm{max}} \sum\nolimits_{i=0}^{t} \lambda^{t-i} \left\| {\boldsymbol{g}}_m \left( \boldsymbol{\theta}_{i} \right) \right\|.     
\end{align}
Using the above inequality in \eqref{EQAppLemVarianceSigmaOmega_1} yields
\begin{align}\label{EQAppLemVarianceSigmaOmega_6}
\sigma_{\omega} (t) \le \frac{\sigma}{\sqrt{P_t}} \bigg( \sum\limits_{m=1}^{M} \frac{1}{\sigma_{\rm{max}} \sum\nolimits_{i=0}^{t} \lambda^{t-i} \left\| {\boldsymbol{g}}_m \left( \boldsymbol{\theta}_{i} \right) \right\| + 1} \bigg)^{-1}.    
\end{align}
For $a_1, ..., a_n \in \mathbb{R}^+$, we prove in Appendix \ref{AppArithmaticMeanInverseMean} that 
\begin{align}\label{ArithmaticMeanInverseMean}
\frac{1}{\sum\nolimits_{i=1}^{n} \frac{1}{a_i}} \le \frac{\sum\nolimits_{i=1}^{n} a_i}{n^2}.     
\end{align}
According to \eqref{ArithmaticMeanInverseMean}, we can rewrite \eqref{EQAppLemVarianceSigmaOmega_6} as follows:
\begin{align}\label{EQAppLemVarianceSigmaOmega_7}
\sigma_{\omega} (t) &\le \frac{\sigma}{M^2 \sqrt{P_t}}  \sum\nolimits_{m=1}^{M} \left( \sigma_{\rm{max}} \sum\nolimits_{i=0}^{t} \lambda^{t-i} \left\| {\boldsymbol{g}}_m \left( \boldsymbol{\theta}_{i} \right) \right\| + 1 \right).    
\end{align}
Independence of the gradients across time and devices yields
\begin{align}\label{EQAppLemVarianceSigmaOmega_8}
\mathbb{E} \left[ \sigma_{\omega} (t) \right] & \le  \frac{\sigma}{M^2 \sqrt{P_t}}  \bigg( \sigma_{\rm{max}} \sum\limits_{i=0}^{t} \lambda^{t-i}  \sum\limits_{m=1}^{M} \mathbb{E} \left[ \left\| {\boldsymbol{g}}_m \left( \boldsymbol{\theta}_{i} \right) \right\| \right] + M \bigg)  \nonumber\\
& \le \frac{\sigma}{M \sqrt{P_t}} \left( \sigma_{{\rm{max}}} \bigg( \frac{1-\lambda^{t+1}}{1-\lambda} \bigg) G + 1 \right).    
\end{align}

\section{Proof of Inequality \eqref{ArithmaticMeanInverseMean}}\label{AppArithmaticMeanInverseMean}
To prove the inequality in \eqref{ArithmaticMeanInverseMean}, we need to show that, for $a_1, ..., a_n >0$,
\begin{align}\label{AppArithmaticMeanInverseMean_1}
\frac{n^2\prod\nolimits_{i=1}^{n}a_i}{\sum\nolimits_{i=1}^{n} \prod\nolimits_{j=1, j \ne i}^{n}a_j} \le \sum\limits_{i=1}^{n} a_i,    
\end{align}
which  corresponds to 
\begin{align}\label{AppArithmaticMeanInverseMean_2}
\sum\limits_{i=1}^{n} \prod\limits_{j=1, j \ne i}^{n}a_j \sum\limits_{k=1}^{n} a_k - n^2\prod\limits_{i=1}^{n}a_i \ge 0.
\end{align}
We have 
{\small{\begin{align}\label{AppArithmaticMeanInverseMean_3}
& \sum\nolimits_{i=1}^{n} \prod\nolimits_{j=1, j \ne i}^{n}a_j \sum\nolimits_{k=1}^{n} a_k - n^2\prod\nolimits_{k=1}^{n}a_k \nonumber\\
& = \sum\limits_{i=1}^{n} \left( \prod\nolimits_{j=1}^{n}a_j + \sum\nolimits_{j=1, j\ne n}^{n} a_j^2 \prod\nolimits_{k=1, k \ne i, j}^{n}a_k  \right) - n^2\prod\limits_{k=1}^{n}a_k \nonumber\\
& = \sum\nolimits_{i=1}^{n} \sum\nolimits_{j=1, j\ne n}^{n} a_j^2 \prod\nolimits_{k=1, k \ne i, j}^{n}a_k - n(n-1) \prod\nolimits_{k=1}^{n}a_k\nonumber\\
& = \sum\nolimits_{(i,j)\in [n]^2, i \ne j} (a_i^2 + a_j^2) \prod\nolimits_{k=1, k \ne i, j}^{n}a_k - n(n-1) \prod\nolimits_{k=1}^{n}a_k\nonumber\\
& = \sum\nolimits_{(i,j)\in [n]^2, i \ne j} (a_i^2 + a_j^2) \prod\nolimits_{k=1, k \ne i, j}^{n}a_k \nonumber\\
& \qquad \qquad \qquad \qquad \qquad \quad - 2 \sum\nolimits_{(i,j)\in [n]^2, i \ne j} a_i a_j \prod\nolimits_{k=1, k \ne i, j}^{n} a_k\nonumber\\
& = \sum\nolimits_{(i,j)\in [n]^2, i \ne j} (a_i^2 + a_j^2 - 2a_ia_j) \prod\nolimits_{k=1, k \ne i, j}^{n}a_k\nonumber\\
& = \sum\nolimits_{(i,j)\in [n]^2, i \ne j} (a_i - a_j)^2 \prod\nolimits_{k=1, k \ne i, j}^{n}a_k \ge 0,
\end{align}}}
which completes the proof of inequality \eqref{ArithmaticMeanInverseMean}.

\section{Proof of Lemma \ref{LemExpectationDifferenceTheta}}\label{AppLemExpectationDifferenceTheta}
By taking the expectation with respect to the gradients, we have
{\small
\newcommand\firstequalityFour{\mathrel{\overset{\makebox[0pt]{\mbox{\normalfont\tiny\sffamily (a)}}}{=}}}
\newcommand\firstinequalityFour{\mathrel{\overset{\makebox[0pt]{\mbox{\normalfont\tiny\sffamily (b)}}}{\le}}}
\newcommand\secondinequalityFour{\mathrel{\overset{\makebox[0pt]{\mbox{\normalfont\tiny\sffamily (c)}}}{\le}}}
\begin{align}\label{EQAppLemExpectationDifferenceTheta}
& \mathbb{E} \left[ \left\| \eta \frac{1}{M} \sum\limits_{m=1}^{M} \left( \boldsymbol{g}_m \left( \boldsymbol{\theta}_t \right)  - {\rm{sp}}_k \left( \boldsymbol{g}_m \left( \boldsymbol{\theta}_t \right) + {\boldsymbol{\Delta}_{m} (t)} \right) \right) - \sigma_{\omega}(t) \boldsymbol{w}(t) \right\| \right] \nonumber\\
& \le \eta \mathbb{E} \left[ \left\| \frac{1}{M} \sum\limits_{m=1}^{M} \left( \boldsymbol{g}_m \left( \boldsymbol{\theta}_t \right) + {\boldsymbol{\Delta}_{m} (t)}  - {\rm{sp}}_k \left( \boldsymbol{g}_m \left( \boldsymbol{\theta}_t \right) + {\boldsymbol{\Delta}_{m} (t)} \right) \right) \right\| \right] \nonumber\\
& \; \; \;\; + \eta \mathbb{E} \left[ \left\| \frac{1}{M} \sum\nolimits_{m=1}^{M} {\boldsymbol{\Delta}_{m} (t)} \right\| \right] + \eta \left\| \boldsymbol{\omega} (t) \right\| \mathbb{E} \left[ \sigma_{\omega}(t) \right] \nonumber\\
& \firstequalityFour \eta \mathbb{E} \left[ \left\| \frac{1}{M} \sum\nolimits_{m=1}^{M} {\boldsymbol{\Delta}_{m} (t+1)} \right\| \right] + \eta \mathbb{E} \left[ \left\| \frac{1}{M} \sum\nolimits_{m=1}^{M} {\boldsymbol{\Delta}_{m}  (t)} \right\| \right] \nonumber\\
& \;\;\;\; + \eta \left\| \boldsymbol{\omega} (t) \right\| \mathbb{E} \left[ \sigma_{\omega}(t) \right]\nonumber\\
& \firstinequalityFour \eta \mathbb{E} \left[ \frac{1}{M} \sum\nolimits_{m=1}^{M} \sum\nolimits_{i=0}^{t} \lambda^{t+1-i} \left\| \boldsymbol{g}_m \left( \boldsymbol{\theta}_i \right) \right\| \right] \nonumber\\
& \;\;\;\; + \eta \mathbb{E} \left[ \frac{1}{M} \sum\nolimits_{m=1}^{M} \sum\nolimits_{i=0}^{t-1} \lambda^{t-i} \left\| \boldsymbol{g}_m \left( \boldsymbol{\theta}_i \right) \right\| \right]\nonumber\\
& \;\;\;\; + \eta \left\| \boldsymbol{\omega} (t) \right\| \mathbb{E} \left[ \sigma_{\omega}(t) \right],\nonumber\\
& \secondinequalityFour \eta \lambda \left(  \frac{(1+ \lambda)(1-\lambda^t)}{1-\lambda} +1 \right)G \nonumber\\
& \;\;\;\; + \eta \rho(\delta) \frac{\sigma}{M \sqrt{P_t}} \left( \sigma_{{\rm{max}}} \bigg( \frac{1-\lambda^{t+1}}{1-\lambda} \bigg) G + 1 \right),
\end{align}}
where (a) follows from the definition $\boldsymbol{\Delta}_m(t)$ in \eqref{AccumErrorUpdateItet}, (b) follows from the inequality in \eqref{EQAppLemVarianceSigmaOmega_4}, and (c) follows from Assumption \ref{AssumpVarianceGrads}, as well as the result presented in Lemma \ref{LemVarianceSigmaOmega}.

\section{Proof of Theorem \ref{TheoremMainProbError}}\label{AppTheoremMainProbError}
We define the following process for the A-DSGD algorithm with the model update estimated as \eqref{ModelUpdateAfterAssumpLem}, $t \ge 0$, 
\begin{align}\label{AppVtDef}
V_t (\boldsymbol{\theta}_t, ..., \boldsymbol{\theta}_0) \triangleq W_t (\boldsymbol{\theta}_t, ..., \boldsymbol{\theta}_0) - \eta L \sum\nolimits_{i=0}^{t-1} v(i),    
\end{align}
if $\boldsymbol{\theta}_i \notin \cal S$, $\forall i \le t$, and $V_t (\boldsymbol{\theta}_t, ..., \boldsymbol{\theta}_0) = V_{\tau_2-1} (\boldsymbol{\theta}_{\tau_2-1}, ..., \boldsymbol{\theta}_0)$, where ${\tau_2}$ denotes the smallest iteration index such that $\boldsymbol{\theta}_{{\tau_2}} \in \cal S$. When the algorithm has not succeeded to the success region, we have, $\forall t \ge 0$,
\begin{align}\label{AppV_tPlus1Inst}
&V_{t+1} (\boldsymbol{\theta}_{t+1}, ..., \boldsymbol{\theta}_0) \nonumber\\
& = W_{t+1} \bigg( \boldsymbol{\theta}_t - \eta \bigg(\frac{1}{M} \sum\limits_{m=1}^{M} \boldsymbol{g}_m^{sp} \left(\boldsymbol{\theta}_t \right) + \sigma_{\omega} (t) \boldsymbol{w} (t) \bigg), \boldsymbol{\theta}_t, ..., \boldsymbol{\theta}_0 \bigg)\nonumber\\
& \;\;\;\; - \eta L \sum\nolimits_{i=0}^{t} v(i)\nonumber\\
& \le W_{t+1} \left(\boldsymbol{\theta}_t - \eta \frac{1}{M} \sum\nolimits_{m=1}^{M} \boldsymbol{g}_m \left(\boldsymbol{\theta}_t \right), \boldsymbol{\theta}_t, ..., \boldsymbol{\theta}_0 \right) \nonumber\\
& \;\;\;\; + L \bigg\| \eta \frac{1}{M} \sum\limits_{m=1}^{M} \left( \boldsymbol{g}_m \left( \boldsymbol{\theta}_t \right)  - \boldsymbol{g}_m^{\rm{sp}} \left( \boldsymbol{\theta}_t \right) \right) - \sigma_{\omega}(t) \boldsymbol{w}(t) \bigg\| \nonumber\\
& \;\;\;\; - \eta L \sum\limits_{i=0}^{t} v(i),
\end{align}
where the inequality is due to the $L$-Lipschitz property of process $W_t$ in its first coordinate. By taking expectation of the terms on both sides of the inequality in \eqref{AppV_tPlus1Inst} with respect to the randomness of the gradients, it follows that 
\newcommand\firstinequalityFive{\mathrel{\overset{\makebox[0pt]{\mbox{\normalfont\tiny\sffamily (a)}}}{\le}}}
\begin{align}\label{AppV_tPlus1Expect}
&\mathbb{E} \left[ V_{t+1} (\boldsymbol{\theta}_{t+1}, ..., \boldsymbol{\theta}_0) \right] \nonumber\\
& \le \mathbb{E} \left[ W_{t+1} \left(\boldsymbol{\theta}_t - \eta \frac{1}{M} \sum\nolimits_{m=1}^{M} \boldsymbol{g}_m \left(\boldsymbol{\theta}_t \right), \boldsymbol{\theta}_t, ..., \boldsymbol{\theta}_0 \right) \right] \nonumber\\
& \;\;\;\;+ \eta L \mathbb{E} \left[ \Big\| \frac{1}{M} \sum\nolimits_{m=1}^{M} \left( \boldsymbol{g}_m \left( \boldsymbol{\theta}_t \right)  - \boldsymbol{g}_m^{\rm{sp}} \left( \boldsymbol{\theta}_t \right) \right) - \sigma_{\omega}(t) \boldsymbol{w}(t) \Big\| \right] \nonumber\\
& \;\;\;\;- \eta L \mathbb{E} \left[ \sum\nolimits_{i=0}^{t} v(i) \right] \nonumber\\
& \firstinequalityFive W_{t} \left(\boldsymbol{\theta}_t, ..., \boldsymbol{\theta}_0 \right) + \eta L v(t) - \eta L \sum\limits_{i=0}^{t} v(i) = V_{t} (\boldsymbol{\theta}_{t}, ..., \boldsymbol{\theta}_0), 
\end{align}
where (a) follows since $W_t$ is a rate supermartingale process, as well as the result in Lemma \ref{LemExpectationDifferenceTheta}. Accordingly, we have
\begin{align}\label{AppVtPlus1VtIneq}
\mathbb{E} \left[ V_{t+1} (\boldsymbol{\theta}_{t+1}, ..., \boldsymbol{\theta}_0) \right] \le V_{t} (\boldsymbol{\theta}_{t}, ..., \boldsymbol{\theta}_0), \quad \forall t.  
\end{align}

We denote the complement of $E_t$ by $E_t^c$, $\forall t$. With the expectation with respect to the gradients, we have
\newcommand\fifthinequal{\mathrel{\overset{\makebox[0pt]{\mbox{\normalfont\tiny\sffamily (a)}}}{\ge}}}
\newcommand\sixthinequal{\mathrel{\overset{\makebox[0pt]{\mbox{\normalfont\tiny\sffamily (b)}}}{\ge}}}
\begin{align}\label{AppProbErrorProof}
\mathbb{E} \left[ W_{0} \right] &= \mathbb{E} \left[ V_{0} \right] \fifthinequal \mathbb{E} \left[ V_{T} \right] \nonumber\\
& = \mathbb{E} \left[ V_{T} | E_T \right] {\rm{Pr}} \{ E_T \} + \mathbb{E} \left[ V_{T} | E^c_T \right] {\rm{Pr}} \{ E^c_T \} \nonumber\\
& \ge \mathbb{E} \left[ V_{T} | E_T \right] {\rm{Pr}} \{ E_T \} \nonumber\\
& = \left( \mathbb{E} \left[ W_{T} | E_T \right] - \eta L \sum\nolimits_{i=0}^{T-1} v(i) \right) {\rm{Pr}} \{ E_T \}\nonumber\\
& \sixthinequal \left( T - \eta L \sum\nolimits_{i=0}^{T-1} v(i) \right) {\rm{Pr}} \{ E_T \},
\end{align}
where (a) follows form \eqref{AppVtPlus1VtIneq}, and (b) follows since $W_t$ is rate supermartingale. Thus, for $\eta$ bounded as in \eqref{EtaBoundTheorem}, we have
\begin{align}
{\rm{Pr}} \{ E_T \} \le \frac{\mathbb{E} \left[ W_{0} \right]}{T - \eta L \sum\nolimits_{i=0}^{T-1} v(i)}.     
\end{align}
Replacing $W_{0}$ from \eqref{ProcessW_tHasNotSucceesed} completes the proof of Theorem \ref{TheoremMainProbError}.

\bibliographystyle{IEEEtran}
\bibliography{Report}

\end{document}